\documentclass[11pt,a4paper]{article}

\usepackage[verbose=true]{microtype}
\usepackage[margin=1in]{geometry}
\usepackage[sort&compress]{natbib}
\usepackage{amsmath}
\usepackage{amssymb}
\usepackage{mathtools}
\usepackage{authblk}
\usepackage[perpage]{footmisc}
\usepackage{paralist}
\usepackage{quantum}
\usepackage{color}
\usepackage{algorithm}
\usepackage{algpseudocode}
\usepackage[colorlinks,citecolor=blue,linktocpage,breaklinks,hypertexnames=false,pdfpagelabels]{hyperref}
\usepackage[amsmath,thmmarks,hyperref]{ntheorem}
\usepackage[capitalise]{cleveref}

\newcommand{\hide}[1]{}

\newtheorem{theorem}{Theorem}
\newtheorem{lemma}[theorem]{Lemma}
\newtheorem{corollary}[theorem]{Corollary}
\newtheorem{proposition}[theorem]{Proposition}
\newtheorem{definition}[theorem]{Definition}

\newtheorem{conjecture}{Conjecture}

\theorembodyfont{\normalfont}
\newtheorem{example}[theorem]{Example}

\theoremstyle{break}
\newtheorem{theorem-break}[theorem]{Theorem}
\newtheorem{lemma-break}[theorem]{Lemma}
\newtheorem{corollary-break}[theorem]{Corollary}
\newtheorem{definition-break}[theorem]{Definition}

\theoremstyle{nonumberplain}
\theorembodyfont{\normalfont}
\newtheorem{proofpart}{Proof}
\theoremsymbol{$\Box$}
\newtheorem{proof}{Proof}

\algblockx[forall]{Forall}{EndForall}%
  [1]{\textbf{for all} #1 \textbf{do}}{\textbf{end}}
\algblockx[loop]{Loop}{EndLoop}%
  [1]{\textbf{loop} #1}{\textbf{end loop}}

\crefrangelabelformat{equation}{\textup{(#3#1#4)}--\textup{(#5#2#6)}}
\Crefname{algorithm}{Algorithm}{Algorithms}
\Crefname{conjecture}{Conjecture}{Conjectures}
\Crefname{remark}{Remark}{Remark}
\Crefname{theorem-break}{Theorem}{Theorems}
\Crefname{corollary-break}{Corollary}{Corollaries}
\Crefname{definition-break}{Definition}{Definitions}
\crefname{part}{part}{parts}  
\Crefname{part}{Part}{Parts}
\crefname{equations}{Eqs.}{Eqs.}
\Crefname{equations}{Equations}{Equations}
\creflabelformat{equations}{\textup{(#2#1#3)}}
\crefrangelabelformat{equations}{\textup{(#3#1#4)}--\textup{(#5#2#6)}}

\newcommand{\affilcr}{\protect\\}

\DeclareMathOperator{\poly}{poly}
\DeclareMathAlphabet{\mathitbf}{OML}{cmm}{b}{it}

\newcommand{\keyword}[1]{\emph{#1}}

\renewcommand{\emptyset}{\varnothing}
\newcommand{\chan}{\mathcal{E}}
\newcommand{\meas}{\mathcal{M}}
\newcommand{\kQSAT}{$k$\nobreakdash--\textsc{QSAT}}
\newcommand{\kSAT}{$k$\nobreakdash--\textsc{SAT}}
\newcommand{\kCNF}{$k$\nobreakdash--CNF}
\newcommand{\taucheck}{$\tau$\nobreakdash-check}

\definecolor{dark-green}{rgb}{0,0.4,0}

\makeatletter
\newcommand{\TODO}[1]{%
  \begingroup%
    \def\@tmpa{#1}%
    \ifx\@tmpa\@empty%
      \color{red}\textup{\textrm{[TODO]}}%
    \else%
      \color{red}\textup{\textrm{[TODO: #1]}}%
    \fi%
  \endgroup
  }
\makeatother

\begin{document}

\title{A constructive commutative\\quantum Lov\'asz Local Lemma, and beyond}

\author[1]{Toby S. Cubitt\thanks{tcubitt@mat.ucm.es}}

\author[2]{Martin Schwarz\thanks{m.schwarz@univie.ac.at}}

\affil[1]{Departamento de An\'alisis Matem\'atico,
  Facultad de CC Matem\'aticas,\affilcr
  Universidad Complutense de Madrid,
  Plaza de Ciencias 3,
  Ciudad Universitaria,
  28040 Madrid, Spain}

\affil[2]{Faculty of Physics, University of Vienna
  Boltzmanngasse 7, A-1090 Vienna, Austria}

\date{6 December 2011}

\maketitle

\begin{abstract}
  The recently proven Quantum Lov\'asz Local Lemma generalises the well-known
  Lov\'asz Local Lemma. It states that, if a collection of subspace
  constraints are ``weakly dependent'', there necessarily exists a state
  satisfying all constraints. It implies e.g.\ that certain instances of the
  \kQSAT{} quantum satisfiability problem are necessarily satisfiable, or that
  many-body systems with ``not too many'' interactions are always
  frustration-free.

  However, the QLLL only asserts existence; it says nothing about how to
  \emph{find} the state. Inspired by Moser's breakthrough classical results,
  we present a constructive version of the QLLL in the setting of commuting
  constraints, proving that a simple quantum algorithm converges efficiently
  to the required state. In fact, we provide two different proofs, one using a
  novel quantum coupling argument, the other a more explicit combinatorial
  analysis. Both proofs are independent of the QLLL. So these results also
  provide independent, constructive proofs of the commutative QLLL itself, but
  strengthen it significantly by giving an efficient algorithm for finding the
  state whose existence is asserted by the QLLL. We give an application of the
  constructive commutative QLLL to convergence of CP~maps.

  We also extend these results to the non-commutative setting. However, our
  proof of the general constructive QLLL relies on a conjecture which we are
  only able to prove in special cases.
\end{abstract}

\pagebreak
\tableofcontents

\section{Introduction and Background}
The Lov\'asz Local Lemma (LLL), proven by \citet{LLL}, is a well-known and
widely-used result in probability theory. It states that if individual events
are ``not too'' dependent on each other and occur with ``not too high'' a
probability, then there is a positive probability that none of them occur.
This is a non-trivial extension of the trivial fact that, if the individual
events were completely independent, and if none of them occurred with
certainty, then there would be a positive probability that none of them occur.
The LLL gives tight bounds on just how dependent the events are allowed to be
and how high the probability of the events can be. In its most general form,
it states:
\begin{theorem}[\citet{LLL}]\label{thm:LLL}
  Let $A_1,A_2,\dots,A_m$ be events in a probability space, and let
  $\Gamma(A_i)$ denote the set of events that are \emph{not} independent of
  $A_i$, excluding $A_i$ itself. If there exist values $0 \leq
  x_1,x_2,\dots,x_m \leq 1$ such that
  \begin{equation}
    \forall i: \Pr(A_i) \leq x_i\cdot\prod_{\mathclap{A_j\in \Gamma(A_i)}} (1-x_j),
  \end{equation}
  then the probability that \emph{none} of the events occur is at least
  $\prod_i(1-x_i)$. (In particular, it is positive.)
\end{theorem}

Applications of the LLL abound \citep{Alon+Spencer}. It is frequently invoked
in order to prove existence of some mathematical object via the probabilistic
method. For example, it can be used to prove existence of solutions to boolean
satisfiability problems. The \kSAT{} problem asks whether there exists an
assignment of truth values to a set of boolean variables that satisfies a
boolean expression in conjunctive-normal form (CNF), where each clause in the
CNF formula involves at most $k$ variables. Letting
$A_i$ be the event that the $i$\textsuperscript{th} clause is violated when
the assignment is chosen at random, the LLL implies:
\begin{corollary}[Symmetric Lov\'asz Local Lemma]
  \label{cor:symmetric_LLL}
  If each variable in a \kCNF{} formula appears in at most $2^{k}/(e\cdot k)$
  clauses, then there exists a satisfying assignment for the formula.
\end{corollary}
Indeed, this amounts to a symmetric, uniform version of the full Lov\'asz
Local Lemma.

Recently, by replacing events with subspaces, and probabilities with relative
dimensions, \citet*{QLLL} succeeded in generalising the LLL to the quantum
setting. The relative dimension $R(X)$ of a subspace $X$ in a vector space $V$
is the ratio $R(X) = \dim(X)/\dim(V)$. Two subspaces $X_i,X_j$ are said to be
\keyword{R-independent} if $R(X_i\cap X_j) = R(X_i)R(X_j)$. In its most
general form, the \keyword{Quantum Lov\'asz Local Lemma} (QLLL) states:
\begin{theorem-break}[\citet*{QLLL}] \label{thm:QLLL}
  Let $X_1,X_2,\dots,X_m$ be subspaces, and let $\Gamma(X_i)$ denote the set
  of subspaces that are \emph{not} R-independent of $X_i$, excluding $X_i$
  itself. If there exist values $0 \leq x_1,x_2,\dots,x_m \leq 1$ such that
  the relative dimensions $R(X_i)$ satisfy
  \begin{equation}\label{eq:quantum_Lovasz_conditions}
    R(X_i) \geq 1 - x_i\cdot\prod_{\mathclap{X_j\in \Gamma(X_i)}} (1-x_j),
  \end{equation}
  then $R(\bigcap_i X_i) \geq \prod_i (1-x_i)$. In particular, the
  intersection $\bigcap_iX_i$ has positive dimension.
\end{theorem-break}
(Note that, because of the properties of R-independence under orthogonal
complement, the QLLL has to be stated ``the other way around'' to the LLL, so
that the subspaces correspond to events that one \emph{does} want to occur.
See \citet*[Lemma~11]{QLLL}, and the discussion thereafter.)

Viewed from one perspective, the QLLL has little to do with quantum physics;
rather, it is a mathematical generalisation of the standard LLL to a
geometrical setting. However, viewed from another perspective, the quantum
version is closely related to current topics of physics research.

For example, just as the LLL can be applied to \kSAT{} problems, the QLLL can
be applied to its quantum generalisation, \kQSAT{} \citep{QSAT}. Boolean
variables become qubits, clauses in a CNF formula become projectors $\Pi_i$
that act non-trivially on $k$ qubits,
and the \kQSAT{} problem asks whether there is a state $\ket{\psi}$ of the
qubits satisfying $\forall i: \Pi_i\ket{\psi} = 0$. The QLLL implies:
\begin{corollary-break}[\citet*{QLLL}]
  Let $\Pi_1,\Pi_2,\dots,\Pi_m$ be a \kQSAT{} instance where all projectors
  have rank at most $r$. If each qubit occurs in at most $2^{k}/(e\cdot r
  \cdot k)$ projectors, then there exists a satisfying state for the \kQSAT{}
  instance.
\end{corollary-break}

An equivalent way of expressing the \kQSAT{} problem is to ask whether the
Hamiltonian $H = \frac{1}{m}\sum_i \Pi_i$ has a zero-energy ground state.
Since the $\Pi_i$ are positive-semidefinite, this is equivalent to asking
whether the ground state of the overall Hamiltonian is simultaneously the
ground state of all the individual local terms; i.e.\ we are asking whether or
not the ground state is \keyword{frustrated}. Replacing the projectors $\Pi_i$
with positive-semidefinite local Hamiltonian terms $h_i$ whose support
(coimage) is $\Pi_i$ does not affect whether the ground state is frustrated.
So the \kQSAT{} problem amounts to asking whether the ground state of an
interacting many-body system is frustrated or not. The QLLL implies that
many-body systems in which each particle interacts with ``not too many''
others are \emph{never} frustrated.

The LLL and its quantum counterpart assert existence, e.g.\ of a satisfying
assignment. But they give no indication as to how to \emph{find} this
assignment. In a breakthrough result, \citet{Moser} gave a beautiful proof of
a constructive version of the classical LLL. This not only gives an
independent proof of the LLL, but does so by providing an efficient algorithm
for \emph{finding} the point in probability space whose existence is asserted
by the LLL (e.g.\ the satisfying assignment to a \kSAT{} problem).

\citet{Moser} originally proved this for the symmetric LLL, achieving the
tight asymptotic scaling but not quite achieving the tight constant in the
bounds of \cref{cor:symmetric_LLL}. With Tardos, he quickly generalised his
proof to cover the general LLL and match the tight constants of \cref{thm:LLL}
\citep{MoserTardos}. The proof imposes a very slight restriction, in requiring
that the events in the LLL be determined by different subsets of underlying,
mutually-independent, random variables:
\begin{theorem}[\citet{MoserTardos}]\label{thm:constructive_LLL}
  Let $p_1,p_2,\dots,p_n$ be mutually-independent random variables, and let
  $A_1,A_2,\dots,A_m$ be events determined by these variables. If there exist
  values $0\leq x_1,x_2,\dots,x_m\leq 1$ such that
  \begin{equation}
    \forall i: \Pr(A_i) \leq x_i\cdot\prod_{\mathclap{A_j\in\Gamma(A)}} (1-x_j),
  \end{equation}
  then there exists an assignment of values to the variables ${p_i}$ such that
  \emph{none} of the events $A_i$ occur. Moreover, there is a randomised
  algorithm (\cref{alg:Moser-Tardos}) that finds this assignment in expected
  time
  \begin{equation}
    O\Bigl(n + \sum_{i=1}^m\frac{x_i}{1-x_i}\cdot\Abs{A_i}\Bigr),
  \end{equation}
  where $\abs{A_i}$ is the number of variables involved in determining event
  $A_i$.
\end{theorem}

The algorithm is almost the simplest randomised algorithm one could imagine
(\cref{alg:Moser-Tardos}). We are looking for an assignment such that none of
the events occur, so we say that an event is \keyword{violated} by an
assignment if it occurs for that assignment. The algorithm maintains a
register $v = v_1v_2\dots v_n$ of assignments to the variables $p_i$. At each
step, it checks whether any event $A_j$ is violated by the current assignment,
and if so replaces the assignments $v_i$ to the variables involved in that
event by values chosen uniformly at random. It repeats this procedure until no
events are violated. A priori it is not obvious that this process will ever
terminate; \cref{thm:constructive_LLL} proves that it in fact terminates in
linear expected time!

\begin{algorithm}[!hbtp]
  \caption{Classical Solver}\label{alg:Moser-Tardos}
  \begin{algorithmic}[1]
    \Function{\textnormal{\texttt{solve\_lll}}}{$(A_1,A_2,\dots,A_m)$}
      \Forall{$p_i$}
        \State $v_i\gets$ a random evaluation of $p_i$
      \EndForall
      \While{$\exists A_j$ violated by $v$}
        \State pick a violated event $A_j$;
          \label{alg:Moser-Tardos:log}
        \Forall{$p_i\in A_j$}
          \State $v_i\gets$ a random evaluation of $p_i$;
        \EndForall
      \EndWhile
      \State \textbf{return} $v$;
    \EndFunction
  \end{algorithmic}
\end{algorithm}

\section{Results}
In this paper we prove a constructive version of the commutative case of the
Quantum Lov\'asz Local Lemma.
%
As in the constructive version of the (classical) LLL, we have to impose a
slight restriction, in requiring that the subspaces respect an underlying
tensor product structure. We can, without loss of generality, take the
underlying structure to be the state space of a set of qudits, analogous to
the underlying random variables in \cref{thm:constructive_LLL}. Each subspace
in the QLLL is then defined on some subset of the qudits. (More precisely, it
is the extension of such a subspace to the full Hilbert space of all the
qudits.)

It will be convenient in the constructive QLLL to represent subspaces by
projectors. We define the relative dimension of a projector $\Pi$ to simply be
the relative dimension of the subspace $X$ onto which it projects: $R(\Pi)
\coloneqq \rank(\Pi)/\dim(V) = \dim(X)/\dim(V) = R(X)$. With the restriction
to an underlying tensor-product space, two subspaces are $R$-independent iff
their corresponding projectors act non-trivially on at least one qudit in
common; we say that the projectors \keyword{intersect} in this case.
Conversely, two projectors that act non-trivially on disjoint subsets of
qudits are said to be \keyword{disjoint}. We can simplify the notation by
letting $[i]$ denote the subset of qudits on which projector $\Pi_i$ acts
non-trivially; then $\Pi_i$ and $\Pi_j$ intersect iff
$[i]\cap[j]\neq\emptyset$. We write $\Gamma(\Pi_i)$ for the set of projectors
that intersect with $\Pi_i$, \emph{excluding} $\Pi_i$ itself, and
$\Gamma^+(\Pi_i) = \Gamma(\Pi_i)\cup\{\Pi_i\}$ for the set that
\emph{includes} $\Pi_i$.

\begin{definition}[Lov\'asz conditions]\label{def:Lovasz_conditions}
  Let $\Pi_1,\Pi_2,\dots,\Pi_m$ be projectors that act on subsets of $n$
  qudits. We say that the set of projectors $\{\Pi_i\}$ satisfies the
  \keyword{Lov\'asz conditions} if there exist values $0 \leq
  x_1,x_2,\dots,x_m \leq 1$ such that
  \begin{equation}
    R(\Pi_i) \leq x_i\cdot\prod_{\mathclap{\Pi_j\in \Gamma(\Pi_i)}} (1-x_j).
  \end{equation}
\end{definition}

We prove the following constructive version of the QLLL (\cref{thm:QLLL}) in
the case of commuting projectors:
\begin{theorem}[Constructive Commutative QLLL]\label{thm:constructive_QLLL}
  Let $\Pi_1,\Pi_2,\dots,\Pi_m$ be \emph{mutually commuting} projectors acting
  on subsets of $n$ qudits. If $\{\Pi_i\}$ satisfy the Lov\'asz conditions,
  then there exists a joint state $\rho$ of the qudits such that $\forall i:
  \tr[\Pi_i\rho] = 0$.

  Moreover, there is a quantum algorithm that converges to a state $\rho'$
  such that $\tr[\Pi_i\rho'] \leq \varepsilon$ (or, equivalently, such that
  $\tr[P_0\rho] \geq 1-\varepsilon$ when $P_0$ is the projector onto the subspace
  \mbox{$\vspan\{\ket{\psi}:\forall i\,\Pi_i\ket{\psi}=0\}$}) in time
  \begin{equation}
    \biggOrder{n + \frac{m}{\varepsilon}\sum_{i=1}^m\frac{x_i}{1-x_i}
                   \cdot\bigAbs{[i]}},
  \end{equation}
  where $\Abs{[i]}$ is the number of qudits on which the projector $\Pi_i$
  acts non-trivially.
\end{theorem}
(Note that, because we restrict to subspaces with a tensor product structure,
there is no longer any need to state the constructive QLLL ``the other way
around'' as there was in \cref{thm:QLLL}.) The algorithm is a generalisation
of \cref{alg:Moser-Tardos} to the quantum setting, and is described in
\cref{sec:constructive_proof}.

There are two significant challenges in generalising Moser's constructive LLL
to the quantum setting. First, the state we are trying to construct may be
highly entangled, whereas the algorithm only has access to measurements of the
local projectors $\Pi_i$. Another way of expressing this in terms of the
Hamiltonian $H$ defining a \kQSAT{} instance is that, in the classical
setting, we know in advance in which basis the overall Hamiltonian $H$ is
diagonal---the computational basis---and the local projectors $\Pi_i$ are
local in this same basis. In the quantum setting, the basis which diagonalises
the overall Hamiltonian is only defined globally, and can be highly entangled.
The projectors $\Pi_i$ will not in general be local in this basis, but will
act non-trivially on the entire system. Were we to na\"ively apply the
\citet{MoserTardos} algorithm in this diagonal basis, every time we measured a
projector $\Pi_i$ to be ``violated'' (i.e.\ we obtained the $\Pi_i$ outcome
upon performing the $\{\Pi_i,\id-\Pi_i\}$ measurement), we would have to
discard the entire state and start from scratch. Thus a na\"ive application of
the \citet{MoserTardos} algorithm to the quantum case reduces to picking a
random state, checking if it satisfies a constraint, and, if not, discarding
the entire state and trying again. This certainly cannot find the correct
state efficiently.

Note that this challenge remains just as problematic even if the projectors
$\Pi_i$ commute. In the commutative case of the QLLL, we may know a priori
that there exists a basis which diagonalizes all projectors simultaneously,
but this basis is still only defined globally, the ground state can still be
highly entangled, and the projectors can still be non-local in the diagonal
basis. (Stabiliser states are a simple example of commuting Hamiltonians with
highly entangled ground states~\citep{Nielsen+Chuang}.) Indeed, this and
related questions in the commutative setting have recently been gaining
increasing attention in the context of Hamiltonian
complexity~\citep{AE11,schuch11,BV03}.

The second challenge comes from non-commutativity: quantum states are
disturbed by measurement. The classical algorithm is free to check which
\kSAT{} clauses are currently satisfied, without affecting the current
variable assignment. But quantum mechanically, even if we measure a \kQSAT{}
projector $\Pi_i$ to be ``satisfied'' (i.e.\ we obtain the desired $\id-\Pi_i$
outcome upon performing the $\{\Pi_i,\id-\Pi_i\}$ measurement), the
measurement can disturb the state so as to increase the probability of measing
another $\Pi_j$ to be violated.

Here, we address and give a complete solution to the first of these two
challenges: we prove a constructive version of the commutative QLLL (i.e.\ the
case in which all the $\Pi_i$ commute). A priori, it is not at all clear that
Moser's proof extends even to the commutative quantum case, for the reasons
discussed above. Nonetheless, by extending the proof techniques of
\citet{Moser} and \citet{MoserTardos} in a more subtle way, we \emph{are} able
to prove a constructive version of the commutative QLLL. Moreover, the quantum
algorithm involved in \cref{thm:constructive_QLLL} is just the natural quantum
generalisation of~\citet{MoserTardos}, and almost the simplest imaginable (see
\cref{alg:quantum}). It also coincides with the natural dissipative state
preparation algorithm studied in \citet{VWC09}.

Whilst the algorithm is straightforward, its analysis is not. We provide
\emph{two} different proofs of the constructive commutative QLLL, using two
very different approaches. The first proof, described in the main text,
generalises the probabilistic approach of~\citet{MoserTardos}. The key step in
the proof is the replacement of the classical coupling argument used
by~\citet{MoserTardos}, with a \emph{quantum coupling argument}, which uses a
coupling by entanglement. To our knowledge, this is the first example of a
quantum coupling argument, used as a proof technique to establish convergence
of a quantum stochastic process, which may be of independent interest. Even
though the entanglement is not used directly as a resource by the algorithm,
it is the unique properties of entanglement that allow us to prove via the
quantum coupling that the algorithm can find the correct global basis even
though it has access only to local measurements.

Coupling arguments are long established as a very powerful proof technique in
probability theory, often providing the simplest or even the only proof of
many results~\citep{Lindvall,Thorisson}. Our quantum generalisation of the
coupling method is no exception, providing an elegant and concise proof of the
constructive commutative QLLL of \cref{thm:constructive_QLLL}. But coupling
arguments often seem a little like ``black magic''. In our second proof,
desribed in \cref{sec:combinatorial_proof}, we replace the coupling argument
with a combinatorial proof. Whilst (as is often the case) the combinatorial
argument is significantly more involved than the coupling argument, it is
nonetheless more explicit. It demonstrates how the algorithm can be understood
as a quantum stochastic process produced by iterated measurement, and leads to
interesting new results on such iterated measurement processes.

To generalise these results to the general non-commutative QLLL requires that
we address the second challenge: non-commutativity of quantum measurement, and
the concomitant measurement-disturbance issue. The difficulty here is that
``satisfied'' measurements now play a significant role, and can sometimes make
things worse instead of better for later measurements. Even the order in which
the ``satisfied'' outcomes occur is now significant. We give a simple
non-commutative counter-example which already violates the crucial bounds that
we prove in the commuting case by both coupling and combinatorial arguments.
This suggests that sharper techniques will be required to address this second
challenge, and extend our results to the general, non-commutative setting.

Nonetheless, although we are not able to give a complete proof of a
constructive QLLL in the non-commutative setting, we \emph{can} prove it if we
assume a technical conjecture. We prove the conjecture in certain simple
cases, and it is also supported in more general cases by numerical evidence
(though the numerics we have done are limited). Furthermore, it is likely that
the conjecture \emph{must} hold if the natural quantum generalisation of
Moser's classical algorithm is to work. If our conjecture is false, either
there is no efficient constructive version of the general QLLL, or an entirely
different approach is needed in the non-commutative setting.

The paper is organised as follows. In \cref{sec:constructive_proof} we briefly
review the classical proof of \citet{MoserTardos}, then prove the constructive
commutative Quantum Lov\'asz Local Lemma using a novel quantum coupling
argument. \Cref{sec:combinatorial_proof} contains an alternative,
combinatorial proof of this result. In \cref{sec:CP_map_convergence} we apply
the results of the previous section to bound the convergence time of certain
classes of CP~maps. In \cref{sec:non-commuting} we generalise the results to
the full non-commuting setting, though we have to assume a technical
conjecture which we are currently unable to prove except in some simple
special cases. In \cref{sec:conclusions} we conclude with an outline of
applications of our results to physics and quantum algorithms, some possible
generalisations in the commutative setting, and a discussion of open problems
and potential directions in the non-commutative setting.

\section{A Constructive Proof of the Commutative
  Quantum Lov\'asz Local Lemma}
\label{sec:constructive_proof}
In order to build a constructive proof of the QLLL, we start by adapting the
elegant argument of \citet{MoserTardos} so that it applies to the commutative
quantum case.
Our first step is to generalise the Moser-Tardos algorithm in the obvious way.
The quantum algorithm acts on an \keyword{assignment register} of $n$ qudits,
holding a state (or \keyword{assignment}) $\ket{\alpha}$. (We abuse notation
slightly by using $\ket{\alpha}$ to denote both the register itself and the
state of the register, even though that state need not be pure.) If the
measurement $\{\Pi_i,\id-\Pi_i\}$ is performed on an assignment and yields the
outcome $\Pi_i$, then we say that projector $\Pi_i$ was \keyword{violated}. We
denote the set of qudits on which a projector $\Pi_i$ acts non-trivially by
$[i]$. In another abuse of notation, the $i$\textsuperscript{th} qudit of the
assignment register will be denoted $\ket{\alpha_i}$, even though the reduced
state of that qudit need not be pure. The algorithm will also require a
uniform random source
$P=(\{\ket{0},\ket{1},\dots,\ket{d-1}\},\prob(\ket{i})=1/d)$ that emits qudits
in a random computational basis state. We use the notation
$\ket{\alpha_i}\gets P$ to indicate replacing the contents of the quantum
register $\ket{\alpha_i}$ with a fresh sample from the random source $P$.

\begin{algorithm}[!hbtp]
  \caption{Quantum Solver}\label{alg:quantum}
  \begin{algorithmic}[1]
    \Function{\textnormal{\texttt{solve\_qlll}}}{$(\Pi_1,\Pi_2,\dots,\Pi_m)$}
      \Forall{$i$}
        \State $\ket{\alpha_i} \gets P$
      \EndForall
      \Loop
        \State pick a projector $\Pi_i$ uniformly at random;
        \State measure the projector $\Pi_i$;
         \If{the projector was violated}
          \State append $\Pi_i$ to the execution log;
          \label{alg:quantum:log}
          \Forall{$j\in [i]$}
            \State $\ket{\alpha_j} \gets P$
            \label{alg:resample}
          \EndForall
        \EndIf
      \EndLoop
    \EndFunction
  \end{algorithmic}
\end{algorithm}

Our task is, firstly, to show that \cref{alg:quantum} converges to a state
satisfying the requirements of \cref{thm:constructive_QLLL}, thereby proving
that such a state exists, and, secondly, to prove that this convergence occurs
in polynomial time.

Before adapting the \citet{MoserTardos} proof to the quantum case, let us
review the high-level structure of their argument. Note that, as they proceed,
\cref{alg:Moser-Tardos,alg:quantum} keep a log
(\cref{alg:Moser-Tardos:log,alg:quantum:log}, respectively) of which events or
projectors were violated. For each entry in the log, \citet{MoserTardos}
imagine constructing a ``witness tree'' from all the log data up to that
point. (We will describe the precise procedure for constructing the trees
later, but the details are not important for the overall structure of the
argument.) Each violation adds one entry to the log, and each log entry is
associated with a different witness tree, so the number of distinct witness
trees that can be constructed from---or \keyword{occur in}---the log gives the
number of violations seen by the algorithm.

Therefore, in order to compute the expected number of violations seen by the
algorithm, we would like to compute the probability that a particular witness
tree can occur in the algorithm's log. Of course, it is not at all clear how to
compute this probability directly. So \citet{MoserTardos} use a coupling
argument to relate the behaviour of the random process generated by the
algorithm to a much simpler random process, whose behaviour is more easily
analysed. They call this simple random process ``\taucheck'', because it
executes a random process on a witness tree $\tau$, and outputs either
``pass'' or ``fail'' at the end. (Again, we will describe the \taucheck{}
procedure later, but the details are not important for now.)

Viewed separately, we have two different random processes: the algorithm,
whose behaviour is difficult to analyse, and \taucheck, whose behaviour is
more easily computed. The coupling is established by proving that, if these
two random processes are fed the \emph{same} random source, then their
behaviour becomes correlated in such a way that if a witness tree $\tau$
occurs in the algorithm's log, then \taucheck{} will always output ``pass'',
i.e.
\begin{equation}\label{eq:passes|occurs}
  \Pr(\text{\taucheck{} passes} \mid \text{$\tau$ occurs in the log}) = 1.
\end{equation}
A simple application of Bayes' theorem then implies that
\begin{subequations}\label{eq:Bayes}
  \begin{align}
    \Pr(\text{$\tau$ occurs in the log})
    &=\frac{\Pr(\text{$\tau$ occurs in the log}
            \mid \text{\taucheck{} passes})}
           {\Pr(\text{\taucheck{} passes}
            \mid \text{$\tau$ occurs in the log})}
      \cdot\Pr(\text{\taucheck{} passes})\\
    &\leq \Pr(\text{\taucheck{} passes}),
\end{align}
\end{subequations}
(using \cref{eq:passes|occurs} and the fact that the numerator, being a
probability, is upper-bounded by 1). Thus we can bound the probability of a
particular witness tree occurring in the algorithm's log (which is the same as
the probability of a particular violation occurring) by the probability that
\taucheck{} passes. In order to bound the latter, \citet{MoserTardos} relate
the \taucheck{} procedure to a Galton-Watson branching process,\footnote{The
  Galton-Watson process is a standard random process in probability theory,
  originally derived in order to model extinction of aristocratic surnames in
  family trees, something of great concern to the Victorians in
  19\textsuperscript{th} century England.} which can be analysed
straightforwardly.


\subsection{Witness Trees}\label{sec:witness-trees}
Imagine constructing a tree of violated projectors from the data in the log,
at any point during the algorithm, exactly as in
\citet[Section~2]{MoserTardos}. Denote projectors appearing in the log by
$\Pi^{(i)}$, for $i=0\dots l$ (where $l$ is the number of entries in the log).
To construct a tree starting from entry $n$ in the log, we work backwards from
projector $\Pi^{(n)}$. First, create a root node labelled by $\Pi^{(n)}$. For
each preceding projector $\Pi^{(i<n)}$ in the log, create a new vertex,
labelling it by $\Pi^{(i)}$, and attach it below an existing vertex labelled
by a projector which intersects with $\Pi^{(i)}$. If there is more than one
such vertex, attach it below the one furthest from the root (breaking ties
arbitrarily). If there is no such vertex, simply discard the new vertex.

Note two key observations about the trees constructed by this procedure.
Firstly, projectors located at the same level in the tree can never intersect.
Secondly, if a vertex labelled by $\Pi^{(i)}$ is located at a higher level of
the tree (closer to the root) than a vertex labelled by $\Pi^{(j)}$, and if
projectors $\Pi^{(i)}$ and $\Pi^{(j)}$ intersect, then $\Pi^{(i)}$ must have
occurred later\footnote{Remember that the tree is constructed by reading the
  log \emph{backwards}.} in the log than $\Pi^{(j)}$, i.e.\ $i>j$. Together,
these properties imply that the tree encodes an ordering of the violated
projectors with respect to the partial order defined by projector
intersections.

A tree is called \keyword{proper} if distinct children of the same vertex
always receive distinct labels, and we say that a tree $\tau$ \keyword{occurs}
in the log if $\tau$ can be generated by constructing a tree starting from
some entry of the log. We now show that trees occurring in the log of
\cref{alg:quantum} still satisfy Lemma~2.1 from \citet{MoserTardos},
reinterpreting it for the quantum case in the obvious way:
\begin{lemma}\label{lem:tau-check}
  Let $\tau$ be a fixed tree and $L$ the (random) log produced by
  \cref{alg:quantum}.
  \begin{enumerate}
  \item If $\tau$ occurs in $L$, then $\tau$ is
    proper. \label[part]{part:proper}
  \item The probability that $\tau$ occurs in $L$ is at most $\prod_{v\in
      \tau}\prob[\Pi(v)]$, \label[part]{part:prob}
  \end{enumerate}
  where $\Pi(v)$ is the projector labelling vertex $v$, and $\prob[\Pi(v)]$ is
  the probability of measuring $\Pi(v)$ on the maximally mixed state.
\end{lemma}
\begin{proof}
  \Cref{part:proper} follows immediately from the above observations about the
  way the trees are constructed. (Indeed, there is nothing quantum about
  \cref{part:proper}, so the statement and proof are identical to the
  corresponding part of Lemma~2.1 in \citet{MoserTardos}.)

  To prove \cref{part:prob}, we introduce a new technique, a \emph{quantum
    coupling argument}, to establish a coupling between \cref{alg:quantum} and
  \cref{alg:tau-check}. \Cref{alg:tau-check} is a straightforward quantum
  version of the classical $\tau$-check procedure. In the classical case the
  two processes are coupled by a joint random source. To establish the
  commutative quantum case we replace the joint random source of qudits of
  \cref{alg:quantum} and \cref{alg:tau-check} by a source of maximally
  entangled qudit pairs. One half of the pair is provided to the each of the
  two coupled quantum processes. Since tracing out one half of each pair
  results in the maximally mixed state on the other half, the coupling is
  undetectable by the process acting on the remaining subsystem: the marginal
  distributions of the two coupled random processes are identical to the
  original, separate random processes \cref{alg:quantum} and
  \cref{alg:tau-check}. Nevertheless, entanglement across the subsystem
  boundary can be exploited to bound joint event probabilities, just as
  correlations are exploited by a classical coupling argument. More formally,
  we proceed by proving the following lemmas, which together imply
  \cref{part:prob}.

\begin{algorithm}[!hbtp]
  \caption{Quantum $\tau$-check}\label{alg:tau-check}
  \begin{algorithmic}[1]
    \Function{\textnormal{\taucheck}}{$\tau$}
      \State sort nodes of $\tau$ in reverse breadth-first order;
      \Forall{$v\in\tau$}
        \Forall{$x \in [i]$}
          \State $\ket{\beta_i} \gets Q_i$
        \EndForall
        \State measure $\Pi_i^T$ on $\ket{\beta_{[i]}}$;
        \If{the measurement was satisfied (gave outcome $\id-\Pi_i$)}
          \State \textbf{return} ``failed''
        \EndIf
      \EndForall
      \State \textbf{return} ``passed''
    \EndFunction
  \end{algorithmic}
\end{algorithm}

\begin{lemma}\label{lem:pre-tau-check}
  $\Pr(\text{\taucheck{} passes}) = \prod_{v\in\tau}\Pr[\Pi(v)]$, where
  $\Pr[\Pi(v)]$ is the probability of measuring $\Pi(v)$ on the maximally
  mixed state.
\end{lemma}
\begin{proof}
  This is trivial, since \taucheck{} independently measures each projector
  $\Pi(v)$ on the maximally mixed state drawn from the uniform random sources
  $Q_i$. (Indeed, this is true independent of the order in which \taucheck{}
  visits the vertices.)
\end{proof}

The next lemma establishes the quantum generalisation of
\cref{eq:passes|occurs} in the classical argument.
\begin{lemma}\label{lem:coupling}
  Couple the random sources $P_i$ of \cref{alg:quantum} and $Q_i$ of
  \cref{alg:tau-check}, so that the $n$th qudit from $P_i$ is maximally
  entangled with the $n$th qudit from $Q_i$. Then, if $\{\Pi_i\}$ commute,
  $\Pr(\text{\taucheck{} passes} | \tau\text{ occurs}) = 1$.
\end{lemma}
\begin{proof}
  First, notice that if $\tau$ occurs in \cref{alg:quantum}'s execution log,
  then each label $\Pi(v)$ in $\tau$ corresponds to a unique violation
  $\Pi(v)$ in the log. Furthermore, the partial ordering with respect to
  projector intersections of that section of the log that generates $\tau$, is
  precisely the partial order defined by $\tau$ itself. Therefore, by visiting
  the vertices in reverse-breadth-first order, \taucheck{} will measure the
  projectors $\Pi(v)^T$ in the same order as \cref{alg:quantum} measured
  $\Pi(v)$, up to reorderings of projectors that act on disjoint sets of
  qudits. Thus, when \taucheck{} draws random qudits on which order to measure
  $\Pi(v)^T$, it draws exactly the qudits from $Q_i$ that started off
  maximally entangled with those on which \cref{alg:quantum} measured
  $\Pi(v)$.\footnote{So far, this is identical to the argument in
    \cite[Lemma~2.1]{MoserTardos}. However, between drawing qudits from $P_i$
    and measuring $\Pi(v)$ on them, \cref{alg:quantum} may have performed
    arbitrarily many ``satisfied'' measurements.}

  It is helpful to picture each random source $P_i,Q_i$ as a semi-infinite
  stack of qudits (each being one half of a maximally entangled pair between
  $P_i$ and $Q_i$). When a fresh qudit is drawn, the random source takes it
  from the bottom of it stack. The $n$th qudit in the $P_i$ stack is then in
  one-to-one correspondence with the $n$th qudit in the $Q_i$ stack, with which
  it is maximally entangled. Also, when \cref{alg:quantum} replaces the qudits
  on which it has just measured a violation $\Pi_i$ with fresh ones from the
  $P_i$ stacks (\cref{alg:resample}), we can equivalently think of this as
  appending the fresh qudits to the register $\ket{\alpha}$, and redefining
  all projectors that act on qudits $[i]$ to now act on these new qudits
  instead. For a given state of register $\ket{\alpha}$, let $\ket{\beta}$
  denote the corresponding qudits from the $Q_i$ stacks. (As with
  $\ket{\alpha}$ this is a slight abuse of notation, as the reduced state of
  the qudits in $\ket{\beta}$ will not necessarily be pure.)

  In this picture, let $P_t$ denote the projector on the appropriate qudits on
  $\ket{\alpha}$ corresponding to the $t$th measurement outcome of
  \cref{alg:quantum}, which could be either a violation $\Pi_i$ or a satisfied
  measurement $\id-\Pi_i$. (Since we are now appending qudits from $P_i$ to
  the register $\ket{\alpha}$ whenever \cref{alg:quantum} measures a
  violation, and simply redefining which qudits of $\ket{\alpha}$ future
  measurements $\Pi_i$ act on instead of tracing out the old ones, the set of
  qudits on which the projector $P_t$ acts depends on the violations $P_{i<t}$
  that occurred prior to $P_t$.)

  Let $\ket{\Omega} = \ox\ket{\omega} = (\sum_i\ket{i}\ket{i}/d)^{\ox m}$ denote
  the maximally entangled state between $\ket{\alpha}$ and $\ket{\beta}$.
  Assume \cref{alg:quantum} measures a violation $P_t=\Pi(v)$ (acting on the
  appropriate qudits of the growing $\ket{\alpha}$ register) in the $t$th
  measurement. The (unnormalised) state of $\ket{\alpha,\beta}$ after this
  measurement is then given by
  \begin{equation}
    \ket{\alpha,\beta} =
    \bigl[\Pi(v)P_{t-1}\cdots P_2P_1\ox\id\bigr]\ket{\Omega}
  \end{equation}

  What is probability that \taucheck{} ``fails'' when it performs the
  corresponding $\Pi(v)^T$ measurement, given that $\tau$ occurs in the log?
  We know from above that when \taucheck{} draws qudits from $Q_{[v]}$ to
  measure $\Pi(v)^T$, it obtains precisely those qudits from $\ket{\beta}$
  that correspond to the qudits in $\ket{\alpha}$ on which $\Pi(v)$ acts. Now,
  \taucheck{} outputs ``fail'' if it measures $\id-\Pi(v)^T$, so the
  probability of this occurring is given by
  \begin{subequations}
  \begin{align}
    \tr&\left[\id\ox(\id-\Pi(v)^T) \proj{\alpha,\beta}\right]\\
    &=\BraKet{\alpha,\beta}{\id\ox(\id-\Pi(v)^T)}{\alpha,\beta}\\
    &\propto\BraKet{\Omega}{
        [P_1P_2\cdots P_{t-1}\Pi(v)P_{t-1}\cdots P_2P_1]
        \ox(\id-\Pi(v)^T)}
      {\Omega}\label{eq:propto}\\
    &=\BraKet{\Omega}{
        [P_1P_2\cdots P_{t-1}]
        \ox[(\id-\Pi(v)^T)P_1^TP_2^T\cdots P_{t-1}^T\Pi(v)^T]}
      {\Omega}\label{eq:entangled}\\
    &=\BraKet{\Omega}{
        [P_1P_2\cdots P_{t-1}]
        \ox[P_1P_2^T\cdots P_{t-1}(\id-\Pi(v)^T)\Pi(v)^T]}
      {\Omega}\label{eq:commute}\\
    &= 0,
  \end{align}
  \end{subequations}
  where \cref{eq:propto} is only a proportionality because we have not
  normalised the state, and in \cref{eq:entangled} we have used the property
  of the maximally entangled state $A\ox\id\ket{\omega} = \id\ox
  A^T\ket{\omega}$, which holds for any operator $A$. In \cref{eq:commute}, we
  have used the fact that all the $\{P_i\}$ commute with $\Pi(v)$, which
  follows immediately from the fact that $\{\Pi_i\}$ commute, by assumption.

  This holds for any violation $\Pi(v)$, and any sequence of measurement
  outcomes $P_{i<t}$ prior to $P_t = \Pi(v)$. Therefore, given that $\tau$
  occurs in \cref{alg:quantum}'s log, the probability of \taucheck{}
  ``failing'' on any of its measurements is zero, which proves the lemma.
\end{proof}

If we trace out \taucheck{}, then the states drawn from $P_i$ are maximally
mixed, as required in \cref{alg:quantum}. Similarly, if we trace out the
algorithm, then the states drawn from $Q_i$ by \taucheck{} are maximally mixed
as required in \cref{alg:tau-check}, so \cref{lem:pre-tau-check} still holds when
the random sources are entangled. Thus, as in the analogous classical proof of
\citet{MoserTardos}, \cref{lem:pre-tau-check,lem:coupling} together prove
\cref{part:prob}.
\end{proof}

\subsection{Expected number of violations}\label{sec:expected_violations}
Having shown in \cref{lem:tau-check} that a quantum version of Lemma~2.1 from
\citet{MoserTardos} holds, the remainder of the argument in
\citet{MoserTardos} goes through unchanged. We repeat the argument here for
completeness.

In order to bound the probability of a tree $\tau$ occurring in the log given
in \cref{part:prob} of \cref{lem:tau-check}, \citet{MoserTardos} relate it to
the probability of the following Galton-Watson branching process generating a
proper witness tree $\tau_a$ whose root vertex is labelled by some fixed
projector $\Pi_a$.
In the first iteration, the process produces a root vertex labelled by
$\Pi_a$. In subsequent iterations, the process considers each vertex added in
the previous iteration independently. For vertex $v$ labelled by $\Pi_i$, it
considers each projector $\Pi_j$ in the set $\Gamma^+(\Pi_i)$ of projectors
that intersect with $\Pi_i$ (including $\Pi_i$ itself). Independently for each
such $\Pi_j\in\Gamma^+(\Pi_i)$, it chooses at random whether to add a child
vertex labelled by $\Pi_j$ below $v$ with probability $x_j$, or whether to
skip $\Pi_j$ with probability $1-x_j$. Thus if all the $\Pi_j$ are skipped,
the branch dies out at $v$. The process continues until all branches die out.
(Depending on the probabilities, the process could of course continue
indefinitely.)

Lemma~3.1 of \citet{MoserTardos} computes the probability that this process
produces a given tree $\tau_a$ using standard techniques for Galton-Watson
branching processes. We restate the lemma here for convenience. For notational
simplicity, let
\begin{equation}
  x'_i \coloneqq x_i\;\cdot\;\prod_{\mathclap{\Pi_j\in\Gamma(\Pi_i)}}(1-x_j).
\end{equation}
%
\begin{lemma}[\citet{MoserTardos}]\label{lem:Galton-Watson}
  Let $\tau_a$ be a given proper witness tree whose root vertex is
  labelled by $\Pi_a$. The probability that the Galton-Watson process
  described above produces the tree $\tau_a$ is
  \begin{equation}
    \Pr(\tau_a) = \frac{1-x_a}{x_a}\prod_{\Pi_i\in\tau_a}x'_i,
  \end{equation}
  (where the product is over all vertex labels in $\tau_a$, including
  repetitions).
\end{lemma}

Using this, we can bound the expected number of violations seen by
\cref{alg:quantum}. Note that this bound is \emph{independent} of the number
of iterations in \cref{alg:quantum}. The expected number of violations is
constant even if we run \cref{alg:quantum} forever.
\begin{theorem}\label{thm:expected_violations}
  For any set of mutually-commuting $\{\Pi_i\}$ satisfying the Lov\'asz
  conditions of \cref{def:Lovasz_conditions}, the expected number of
  violations seen by \cref{alg:quantum} is bounded by
  \begin{equation}
    \expectation(\text{\textup{total number of violations}})
    \leq \sum_{i=1}^m \frac{x_i}{1-x_i}.
  \end{equation}
\end{theorem}
\begin{proof}
  Let $N_a$ be the number of times a given projector $\Pi_a$ appears in
  \cref{alg:quantum}'s log, which is the same as the number of times $\Pi_a$
  is measured to be violated. Let $\mathcal{T}_a$ denote all proper witness
  trees whose root vertex is labelled by $\Pi_a$. Then, from
  \cref{lem:tau-check}, we have
  \begin{equation}
    \expectation(N_a)
    = \sum_{\tau\in\mathcal{T}_a}\Pr(\text{$\tau$ appears in the log})
    \leq \sum_{\tau\in\mathcal{T}_a}\prod_{v\in\tau}\Pr[\Pi(v)].
  \end{equation}
  Now, the probability $\Pr[\Pi_i]$ of a projector being violated on a random
  state is just given by its relative dimension $\Pr(\Pi_i) = R(\Pi_i)$. Since
  by assumption the projectors satisfy the Lov\'asz conditions
  (\cref{def:Lovasz_conditions}), the relative dimension satisfies $R(\Pi_i)
  \leq x'_i$. Thus
  \begin{equation}
    \expectation(N_a) \leq \sum_{\tau\in\mathcal{T}_a}\prod_{v\in\tau}\Pr[\Pi(v)]
    \leq \sum_{\tau\in\mathcal{T}_a}\prod_{\Pi_i\in\tau} R(\Pi_i)
    \leq \sum_{\tau\in\mathcal{T}_a}\prod_{\Pi_i\in\tau} x'_i
    \leq \sum_{\tau\in\mathcal{T}_a} \frac{x_a}{1-x_a}\Pr(\tau),
  \end{equation}
  the final inequality following from \cref{lem:Galton-Watson}, $\Pr(\tau_a)$
  being the probability of the Galton-Watson process of
  \cref{lem:Galton-Watson} generating tree $\tau_a$. Since that process either
  produces a proper witness tree in $\mathcal{T}_a$ with root vertex labelled
  by $\Pi_a$, or continues indefinitely, we have
  $\sum_{\tau\in\mathcal{T}_a}\Pr(\tau) \leq 1$, thus
  \begin{equation}
    \expectation(N_a) \leq \frac{x_a}{1-x_a}\sum_{\tau\in\mathcal{T}_a} \Pr(\tau_a)
    \leq \frac{x_a}{1-x_a},
  \end{equation}
  and the theorem follows from summing over all projectors.
\end{proof}

\subsection{An efficient quantum algorithm}
\label{sec:constructive_QLLL_proof}
We are now in a position to prove \cref{thm:constructive_QLLL}: a constructive
version of the commutative Quantum Lov\'asz Local Lemma. From the bound on the
expected number of violations seen by \cref{alg:quantum} already proven in the
previous section, it is easy to prove the existence part of the commutative
Quantum Lov\'asz Local Lemma (i.e.\ that there exists a state satisfying all
the constraints).

\begin{proof}[of existence part of \cref{thm:constructive_QLLL}]
  The theorem asserts existence of a state $\rho$ such that $\forall i:
  \tr[\Pi_i\rho] = 0$, when $\{\Pi_i\}$ is a set of projectors satisfying the
  Lov\'asz conditions. Assume for contradiction that no such state exists.
  Then we have
  \begin{equation}
    \delta \coloneqq \min_{\rho}\frac{1}{m}\sum_i\tr[\Pi_i\rho] > 0.
  \end{equation}
  But $\frac{1}{m}\sum_i\tr[\Pi\rho]$ is just the probability of a randomly
  chosen projector being violated by state $\rho$, so $\delta$ is a lower
  bound on the probability of seeing a violation in one iteration of
  \cref{alg:quantum}. Therefore, the expected number of violations after
  running the algorithm for $t$ iterations is lower-bounded by $t\delta$. For
  sufficiently large $t$, this leads to a contradiction with
  \cref{thm:expected_violations}.
\end{proof}

However, this argument is insufficient to prove the second part of
\cref{thm:constructive_QLLL}: that there exists an \emph{efficient} algorithm
for approximating this state. To prove this, we construct a modified version
of \cref{alg:quantum}, described in \cref{alg:QLLL_converger}, which
essentially runs \cref{alg:quantum} for a random number of iterations.

\begin{algorithm}[!hbtp]
  \caption{Quantum converger}\label{alg:QLLL_converger}
  \begin{algorithmic}[1]
    \Function{\textnormal{\texttt{converge\_qlll}}}{$(\Pi_1,\Pi_2,\dots,\Pi_m)$}
      \State Pick an integer $0 \leq \tau \leq t$ uniformly at random;
      \State Initialise the assignment register as in \cref{alg:quantum};
      \State $i\gets 0$;
      \Loop{$t$ times}
        \If{$i < \tau$}
	  \State Apply one iteration of \cref{alg:quantum}'s main loop;
	  \State $i\gets i+1$;
	\EndIf
      \EndLoop
      \State \textbf{return} assignment register;
    \EndFunction
  \end{algorithmic}
\end{algorithm}

The following general theorem shows that the time required for
\cref{alg:QLLL_converger} to converge to the desired state is directly related
to the expected number of violations seen by \cref{alg:quantum}. (Note that
this theorem does \emph{not} require commuting projectors.)

\begin{theorem}\label{thm:convergence_time}
  For an arbitrary set of $m$ projectors $\{\Pi_i\}$, let $E$ be the expected
  number of violations seen when \cref{alg:quantum} is run with those
  projectors. If $E$ is finite, then the state $\rho$ produced by running
  \cref{alg:QLLL_converger} for time $t = mE/\varepsilon$ satisfies $\forall i:
  \tr[\Pi_i\rho] \leq \varepsilon$.
\end{theorem}
\begin{proof}
  Let $\rho_\tau$ denote the final state of \cref{alg:QLLL_converger}'s
  assignment register given that it picked the value $\tau$. Note that this is
  the same as the state of \cref{alg:quantum}'s assignment register after the
  $\tau$\textsuperscript{th} iteration. Since $\tau$ is chosen at random, and
  we don't learn its value, the final state of \cref{alg:QLLL_converger}'s
  assignment register is described by the density matrix
  \begin{equation}\label{eq:final_state}
    \rho = \frac{1}{t}\sum_{\tau=0}^t \rho_\tau.
  \end{equation}

  Since $\rho_\tau$ is the same as the state of \cref{alg:quantum}'s
  assignment register after $\tau$ iterations, the quantity
  $\frac{1}{m}\sum_i\tr[\Pi_i\rho_\tau]$ is equal to the probability of seeing
  a violation in the $\tau+1$\textsuperscript{th} iteration of
  \cref{alg:quantum}. So the expected number of violations when
  \cref{alg:quantum} runs for $t+1$ iterations is given by
  \begin{equation}\label{eq:tau_expected}
    E = \sum_{\tau=0}^t
        \Pr(\text{violation in $\tau$\textsuperscript{th} iteration})
      = \sum_{\tau=0}^t\biggl(\frac{1}{m}\sum_i\tr[\Pi_i\rho_\tau]\biggr).
  \end{equation}
  This together with \cref{eq:final_state} implies
  \begin{equation}\label{eq:violation_probability}
    \frac{1}{m}\sum_i\tr[\Pi_i\rho]
    =\frac{1}{m}\sum_i\tr\biggl[
      \Pi_i\cdot\frac{1}{t}\sum_{\tau=0}^t\rho_\tau\biggr]
    =\frac{1}{t}\sum_{\tau=0}^t
      \biggl(\frac{1}{m}\sum_i\tr[\Pi_i\rho_\tau]\biggr)
    =\frac{E}{t}.
  \end{equation}
  Since $\tr[\Pi_i\rho]$ are positive, each term must be bounded by
  \begin{equation}
    \tr[\Pi_i\rho] \leq \frac{mE}{t},
  \end{equation}
  thus running for time $t = mE/\varepsilon$ gives the desired bound
  $\tr[\Pi_i\rho] \leq \varepsilon$.
\end{proof}

For commuting $\{\Pi_i\}$, we have a bound $E \leq \sum_ix_i/(1-x_i)$ from
\cref{thm:expected_violations}. Thus
\cref{thm:expected_violations,thm:convergence_time} together prove an
efficient algorithm in the commutative case. Accounting for all parameters in
the run-time analysis, note that \cref{alg:quantum} requires $\order{n}$ time
to generate the initial assignment, and $\order{\abs{[i]}}$ time to resample
$\abs{[i]}$ qudits for each iteration in which projector $\Pi_i$ was violated.
Thus, if we run \cref{alg:QLLL_converger} for $t =
m/\varepsilon\sum_ix_i/(1-x_i)$ iterations, taking a total time of
\begin{equation}
  \biggOrder{n + \frac{m}{\varepsilon}
    \sum_{i=1}^m\frac{x_i}{1-x_i} \bigAbs{[i]}},
\end{equation}
then \cref{thm:convergence_time} guarantees that the resulting state $\rho$
will satisfy $\forall i : \tr[\Pi_i\rho] \leq \varepsilon$.

To see that this is equivalent to the condition $\tr[P_0\rho] \geq
1-\varepsilon$, where $P_0$ is the projector onto the subspace
$\vspan\{\ket{\psi}:\forall i\,\Pi_i\ket{\psi}=0\}$, note that for commuting
projectors $P_0 = \id - \prod_{i=1}^m\Pi_i$. Thus
\begin{equation}
  \tr[P_0\rho] = \tr\biggl[\Bigl(\id - \prod_{i=1}^m\Pi_i\Bigr)\rho\biggr]
  = 1 - \tr\left[\prod_{i=1}^m\Pi_i\rho\Pi_i\right]
  \geq 1 - \tr[\Pi_i\rho]
  \geq 1 - \varepsilon,
\end{equation}
where the second equality relies of the fact that the $\{\Pi_i\}$ mutually
commute.

Thus the state $\rho$ fulfils all the requirements of
\cref{thm:constructive_QLLL}, thereby proving the constructive commutative
Quantum Lov\'asz Local Lemma.

\subsection{An exact efficient quantum algorithm}
\label{sec:eff_QLLL}
The algorithm of the previous section constructs in polyomial-time a state
that is $\varepsilon$-close to the one whose existence is asserted by the
QLLL. This is the natural behaviour to demand of a constructive algorithm in
the quantum setting. As we will see in \cref{sec:non-commuting}, it
generalises straightforwardly to the non-commutative case.

However, \cite{Moser} and \cite{MoserTardos} express the classical result
slightly differently. They prove that their algorithm constructs the desired
state \emph{exactly}, in polynomial \emph{expected}
time.\footnote{\cite{MoserTardos} also prove a deterministic version.} It is
not reasonable to demand this in the general quantum setting. But in the case
of commuting projectors, this is a valid and equivalent way to formulate the
constructive commutative QLLL. In this section, we use the results of the
previous sections to give an algorithm that constructs the desired state
exactly, with polynomial expected run-time.

The subspace of states fulfilling the QLLL requirement $\forall i:
\Pi_i\ket{\psi} = 0$ can be characterised as the intersection of the supports
(coimages) of the projectors $\id-\Pi_i$. Let $S_i$ be the support of
$P_i=\id-\Pi_i$. Then the desired subspace $S_0$ is $S=\bigcap_i^m S_i$, and
$P_0$ is exactly the projector onto $S_0$.
Given any state $\rho$ with non-zero overlap with $P_0$, i.e.\
$\tr[P_0\rho]>0$, we can construct a state in $P_0$ by projecting onto $P_0$:
$P_0\rho P_0/\tr(P_0 \rho)$. For commuting projectors $\Pi_i$,
$P_0=\prod_{i=1}^m (\id-\Pi_i)$. This motivates \cref{alg:exact_QLLL}, a
slight modification of \cref{alg:quantum} where we fix an order in which to
apply the projectors. We immediately observe the following:

\begin{lemma}\label{cor:convergence}
  Let $P_0\rho P_0$ be the \emph{unnormalised} projection of $\rho$ onto the
  subspace of states satisfying the QLLL requirement $\forall i:
  \Pi_i\ket{\psi}=0$. Consider a run of \cref{alg:exact_QLLL}, starting from
  an initial state $\rho$ with $\tr[P_0 \rho] > 0$. If no measurement has been
  violated throughout the entire run, then the assignment register contains
  the state $T \rho T$, where $T=P_m\cdot P_{m-1} \dotsm P_2\cdot P_1$.
\end{lemma}

\begin{algorithm}[!hbtp]
  \caption{Exact Commutative QLLL solver}\label{alg:exact_QLLL}
  \begin{algorithmic}[1]
    \Procedure{\textnormal{solve\_qlll}}{$(\Pi_1,\Pi_2,\dots,\Pi_m)$}
      \Forall{$x\in X$}
        \State $\ket{\alpha_i} \gets P$
      \EndForall
      \State $c \gets 0$
      \Loop{until $c = m$ or for a maximum of $(m+1)(pm'+1)$ iterations}
        \State Pick the next projector $\Pi_i$ according to some fixed order;
        \State Measure $\Pi_i$;
        \If{the projector was violated}
          \State $c \gets 0$;
          \State Append $\Pi_i$ to the execution log;
          \Forall{$x\in [i]$}
            \State $\ket{\alpha_i} \gets P$;
          \EndForall
        \Else
           \State $c \gets c + 1$;
        \EndIf
      \EndLoop
      \If{$c = m$}
        \State Return ``success''
      \Else
        \State Return ``failure''
      \EndIf
    \EndProcedure
  \end{algorithmic}
\end{algorithm}

Quantum mechanics does not allow us to apply projector $T=P_m\cdot P_{m-1}
\dotsm P_2\cdot P_1$ deterministically: some measurements will be violated
with positive probability. But we know from \cref{thm:expected_violations},
which applies equally well to \cref{alg:exact_QLLL}, that the expected number
of violations is constant no matter how many iterations are performed in
\cref{alg:exact_QLLL}. We can use this to upper-bound the expected number of
iterations required to produce one contiguous run of all $m$ projectors.

\begin{lemma}\label{lem:termination}
  For any integer $p>1$, the probability that \cref{alg:exact_QLLL} returns
  ``success'' is at least $1-1/p$. In this case, the assignment register is in
  state $T \rho T /\tr[T \rho T]$, where $T$ is the desired sequence
  of projections.
\end{lemma}
\begin{proof}
  From \cref{thm:expected_violations}, we know that the expected number of
  violations is bounded from above by $m'=\sum_i x_i/(1-x_i)$. Let $M$ be the
  random variable counting the total number of violations. By Markov's
  inequality, $\Pr[M \leq pm'] \geq 1-1/p$. Thus, running the algorithm for
  $(m+1)(pm'+1)$ iterations, we will see with probability $1-1/p$ at most
  $pm'$ undesired projections.

  These $pm'$ projections partition the execution log of length $(m+1)(pm'+1)$
  into at most $pm'+1$ stretches of desired projections. Thus at least one run
  of length $m$ consisting \emph{exclusively} of the desired projections must
  occur within any run of the algorithm with probability $1-1/p$. The
  termination condition of \cref{alg:exact_QLLL} ensures that it halts after
  the first occurence of such a run.
\end{proof}

Combining \cref{lem:termination,cor:convergence}, we see that the algorithm
must perform $\order{pmm'}$ iterations of which at most $pm'$ lead to a
resampling of $p\sum_{i=1}^m\frac{x_i}{1-x_i}\cdot\bigAbs{[i]}$ qudits, thus
leading to an overall effort of
\begin{equation}
  \Order{n+pm\sum_{i=1}^m\frac{x_i}{1-x_i}\cdot\bigAbs{[i]}},
\end{equation}
where $n$ reflects the cost for sampling the initial $n$ random qudits. Note
that we can easily recover \cref{thm:constructive_QLLL} from this, by defining
$p = 1/\varepsilon$.

Assume we repeatedly run \cref{alg:exact_QLLL} until it suceeds. By fixing
$p>1$ to be some constant, we know from \cref{lem:termination} that in each
iteration \cref{alg:exact_QLLL} succeeds with probability $1-1/p$, after which
it stops, or fails with probability $1/p$ leading to another iteration. Thus
this is a process of repeated Bernoulli trials until the first success with a
geometric distribution $\prob(X=k) =\linebreak (1/p)(1-1/p)^{k-1}$, and
expected termination time $\expectation(X)=p$. Thus, this process has an
overall expected run-time of
\begin{equation}
  \Order{n+m\sum_{i=1}^m\frac{x_i}{1-x_i}\cdot\bigAbs{[i]}},
\end{equation}
and terminates only when it has applied a contiguous run of all $m$
mutually-commuting projectors, which constructs the desired state exactly.

\section{A Combinatorial Proof}\label{sec:combinatorial_proof}
In this section, we give an alternative proof of the key
\cref{thm:expected_violations}, which bounded the expected number of
violations seen by \cref{alg:quantum}. In \cref{sec:constructive_proof}, we
proved \cref{thm:expected_violations} using a quantum coupling argument. Here,
we replace the coupling argument with a combinatorial one. The resulting proof
is substantially more involved than the proof based on coupling, which attests
to the power of even the simple quantum coupling argument used in
\cref{sec:constructive_proof}. However, like many such proofs, the quantum
coupling argument may look a little like ``black magic''. The alternative
proof described here uses straightforward linear algebra and a slight
generalisation of witness trees to directed acyclic graphs (DAG) to arrive at
the same bound, providing different insight into how the QLLL algorithm works.

\citet{MoserTardos} prove the constructive LLL by bounding the probability of
observing a given sequence of violated events in \cref{alg:Moser-Tardos}. The
quantum generalisation of this type of algorithm is a quantum stochastic
process generated by iterated measurement. Here, we prove a number of results
concerning iterated quantum measurement processes, which may be of independent
interest. We then re-prove \cref{thm:expected_violations} using these results.

\subsection{A Simple Iterated Measurement Process}
As a warm up, let us start with the following quantum stochastic process,
which will serve as a building block for later, more complex processes.
\begin{lemma}\label{lem:halting_measurement_process}
  Let $\{\Pi_i\}$ be a finite set of commuting projectors on a $\CC^d$.
  Consider the quantum stochastic process produced by repeatedly picking one
  of the projectors uniformly at random, and performing the two-outcome
  measurement $\{\Pi_i,\id-\Pi_i\}$, starting from the maximally mixed state
  $\id/d$, until a measurement gives the $\Pi$ outcome, whereupon the process
  halts.

  Let $\rho_a$ be the final state of the system given that the process halted
  on outcome $\Pi_a$, and $p_a$ the probability that this occurs. Then the
  unnormalised density operator $X_a = p_a\rho_a$ corresponding to this
  outcome satisfies the operator inequality
  \begin{equation}
    X_a = p_a\rho_a \leq \Pi_a/d,
  \end{equation}
  hence
  \begin{equation}
    p_a \leq \tr[\Pi_a/d].
  \end{equation}
\end{lemma}
\begin{proof}
  Let $m$ be the total number of projectors, so that in each iteration there
  is a probability $1/m$ of picking a given projector. The generalised
  measurement describing the different possible outcomes in a given iteration
  is described by a collection of trace non-increasing CP maps $\{\meas_k\}$,
  such that $\sum_k\meas_k$ is trace-preserving. To obtain the outcome $\Pi_a$
  in a given iteration, we must pick the projector $\Pi_a$ in that iteration
  and also obtain the outcome $\Pi_a$ in the resulting $\{\Pi_i,\id-\Pi_i\}$
  measurement. This corresponds to the measurement element
  $\tfrac{1}{m}\meas_a(\rho) = \tfrac{1}{m}\Pi_a\rho\Pi_a$.

  In order to have reached that iteration, we must have obtained one of the
  $\id-\Pi$ outcomes in all previous iterations, otherwise the process would
  already have halted. The measurement element corresponding to the process
  \emph{not} halting at a particular iteration is therefore given by
  $\meas_{\mathrm{cont}}(\rho) =
  \sum_i\frac{1}{m}(\id-\Pi_i)\rho(\id-\Pi_i)$. Thus the unnormalised density
  operator produced by the process halting with outcome $\Pi_a$ in the
  $t+1$\textsuperscript{th} iteration is described by
  \begin{equation}
    \frac{1}{m}\meas_a\circ
      \underbrace{\meas_{\mathrm{cont}}\circ\cdots
      \circ\meas_{\mathrm{cont}}}_t(\rho)
    =\frac{1}{m}\meas_a\circ\meas_{\mathrm{cont}}^t(\id/d).
  \end{equation}
  The process can halt after any number of iterations, so the overall
  unnormalised density operator corresponding to the process eventually
  halting on outcome $\Pi_a$ is
  \begin{subequations}\label{eq:measurement_process}
  \begin{align}
    X_a
    &=\frac{1}{m}\sum_{t=0}^\infty \sum_{i_0,\dots,i_t} \frac{1}{m^t}
        \Pi_a (\id-\Pi_{i_t})\cdots(\id-\Pi_{i_0})
        \frac{\id}{d}\,(\id-\Pi_{i_0})\cdots(\id-\Pi_{i_t}) \Pi_a\\
    &=\frac{1}{m}\sum_{t=0}^\infty \sum_{i_0,\dots,i_t \neq a} \frac{1}{m^t}
        \Pi_a (\id-\Pi_{i_t})\cdots(\id-\Pi_{i_0})
        \frac{\id}{d}\,(\id-\Pi_{i_0})\cdots(\id-\Pi_{i_t}) \Pi_a
        \label{eq:kill_terms}\\
    &\leq\frac{1}{m}\sum_{t=0}^\infty \sum_{i_0,\dots,i_t \neq a} \frac{1}{m^t}
        \Pi_a \frac{\id}{d}\,\Pi_a \\
    &=\frac{1}{m}\sum_{t=0}^\infty \left(\frac{m-1}{m}\right)^t
        \Pi_a \frac{\id}{d}\,\Pi_a \\
    &=\Pi_a/d
  \end{align}
  \end{subequations}
  as claimed. The equality in \cref{eq:kill_terms} comes from commuting
  $\Pi_a$ all the way through the product of $(\id-\Pi_i)$'s, killing all terms
  in the sum containing any $(1-\Pi_a)$ factor. Thus the sum is only over the
  remaining $(m-1)^t$ terms. The inequality follows from the fact that, since
  all the $(\id-\Pi_i)$'s commute,
  \begin{equation}
    (\id-\Pi_{i_t})\cdots(\id-\Pi_{i_0})
      \frac{\id}{d}\,(\id-\Pi_{i_0})\cdots(\id-\Pi_{i_t})
    \leq \frac{\id}{d}.
  \end{equation}

  The final part of the lemma follows immediately from the fact that the
  probability of the process halting on outcome $\Pi_a$ is given by
  $\tr[X_a]$.
\end{proof}

\subsection{Resample DAGs}
We now turn to a more complicated quantum stochastic process that produces an
infinite sequence of measurement outcomes. For this we require some basic
definitions.

\begin{definition}[Resample DAG]
  \label{def:resample_DAG}
  Let $\{\Pi_i\}$ be a finite set of projectors acting on a tensor product
  space $\bigotimes_i\HS_i$, and let $\Pi_{a_1},\Pi_{a_2},\dots,\Pi_{a_l}$ be
  a sequence of projectors chosen from this set. The \keyword{resample DAG}
  $\mathfrak{G}(\Pi_{a_1},\Pi_{a_2},\dots,\Pi_{a_l})$ of such a sequence is
  the directed acyclic graph (DAG) with vertices labelled by sequence
  elements, and a directed edge from $\Pi_{a_i}$ to $\Pi_{a_j}$ iff the two
  projectors intersect and $\Pi_{a_i}$ occurs \emph{before} $\Pi_{a_j}$. (Thus
  the resample DAG encodes the partial ordering of the sequence with respect
  to projector intersections.)
\end{definition}

(Note that this is the same as the \keyword{``resample graph''} defined
by~\cite{KashyapSzegedy} for the classical setting.)
$\mathfrak{G}(\Pi_{a_1},\Pi_{a_2},\dots,\Pi_{a_l})$ can be constructed from
the sequence as follows. Work \emph{backwards} through the sequence, starting
with the final projector $\Pi_{a_l}$. For each projector $\Pi_{a_i}$, find the
set $L_i$ of all vertices labelled by projectors that intersect with
$\Pi_{a_i}$ (which may be empty). Add a new vertex, labelling it by
$\Pi_{a_i}$, and create directed edges from each element of $L_i$ to the new
vertex.

\begin{definition}[DAG Probability]
  \label{def:DAG_probability}
  Let $\{\Pi_i\}$ be a finite set of projectors acting on a tensor product
  space $\bigotimes_i\HS_i$, and let $\mathfrak{G}$ be a resample DAG over
  these projectors. The \keyword{DAG probability} with respect to
  $\mathfrak{G}$ of the sequence $\Pi_{a_1},\Pi_{a_2},\dots,\Pi_{a_l}$,
  denoted $p_\mathfrak{G}(\Pi_{a_1},\Pi_{a_2},\dots,\Pi_{a_l})$, is the
  probability that the following process generates the sequence. Starting with
  the resample DAG $\mathfrak{G}$ and an empty initial sequence, pick one of
  the leaf vertices uniformly at random and append its label to the end of the
  sequence, removing the vertex from the DAG. Repeat this procedure until no
  vertices are left.
\end{definition}

\subsection{A More Complicated Iterated Measurement Process}
We are now in a position to study a more complicated quantum stochastic
process, on which our constructive QLLL algorithm will ultimately be based.

\begin{lemma}\label{lem:perpetual_measurement_process}
  Let $\{\Pi_i\}$ be a finite set of commuting projectors on an $n$-qudit
  Hilbert space $\bigotimes^n\HS$ with total dimension $d$. Consider the
  quantum stochastic process which starts from the maximally mixed state
  $\id/d$, runs the process of \cref{lem:halting_measurement_process} until it
  halts on some measurement outcome $\Pi_a$, then applies the trace-preserving
  CP~map
  \begin{equation}
    \chan_a(\rho) = \tr_{[a]}(\rho)\ox\frac{\id_{[a]}}{d_{[a]}},
  \end{equation}
  which reinitialises the qudits measured by $\Pi_a$ to the maximally mixed
  state, and repeats these two steps indefinitely. Here, $[a]$ denotes the
  subset of qudits on which projector $\Pi_a$ acts non-trivially, and
  $d_{[a]}$ denotes the Hilbert space dimension of that subset.

  Let $X_{a_1,a_2,\dots,a_t}$ denote the unnormalised density operator
  corresponding to obtaining an outcome in which the first $t$ iterations of
  the \cref{lem:halting_measurement_process} sub-process halted in turn on the
  sequence of measurement outcomes $\Pi_{a_1},\Pi_{a_2},\dots,\Pi_{a_t}$. Then
  $X_{a_1,a_2,\dots,a_t}$ satisfies the operator inequality
  \begin{equation}\label{eq:seq_density_operator_inequality}
    X_{a_1,a_2,\dots,a_t}
    \leq p_\mathfrak{G}(\Pi_{a_1},\Pi_{a_2},\dots,\Pi_{a_t})
      \,\prod_{i=1}^t\tr[\Pi_{a_i}/d]\cdot\frac{\id}{d},
  \end{equation}
  and the probability of this occurring is bounded by
  \begin{equation}\label{eq:seq_probability_bound}
    \Pr(\Pi_{a_1},\Pi_{a_2},\dots,\Pi_{a_t})
    \leq p_\mathfrak{G}(\Pi_{a_1},\Pi_{a_2},\dots,\Pi_{a_t})
         \prod_{i=1}^t\tr[\Pi_{a_i}/d],
  \end{equation}
  where the resample DAG $\mathfrak{G} =
  \mathfrak{G}(\Pi_{a_1},\Pi_{a_2},\dots,\Pi_{a_t})$.
\end{lemma}

Before proving this lemma, we first establish some basic identities concerning
compositions of the various maps involved in
\cref{lem:halting_measurement_process,lem:perpetual_measurement_process}.

\begin{lemma}\label{lem:CP_map_identities}
  Let the CP maps $\meas_a$, $\meas_{\mathrm{cont}}$ and $\chan_a$ be defined
  as in
  \cref{lem:halting_measurement_process,lem:perpetual_measurement_process}:
  \begin{gather}
    \meas_a(\rho) = \Pi_a\rho\,\Pi_a,\\
    \meas_{\mathrm{cont}}(\rho) = \sum_i\frac{1}{m}(\id-\Pi_i)\rho(\id-\Pi_i),\\
    \chan_a(\rho) = \tr_{[a]}(\rho)\ox\frac{\id_{[a]}}{d_{[a]}}.
  \end{gather}
  In addition, for \emph{disjoint} $\Pi_a,\Pi_b$, define the CP maps
  \begin{gather}
    \meas_\Sigma = \sum_{t=0}^\infty\meas_{\mathrm{cont}}^t,\\
    \meas_{a,b}(\rho) = (\Pi_a\ox\Pi_b) \rho(\Pi_a\ox\Pi_b),
      \phantom{\sum_t^\infty}\\
    \chan_{a,b}(\rho) = \chan_a\circ\chan_b(\rho)
      = \tr_{[a]\cup[b]}(\rho)\ox\frac{\id_{[a]\cup[b]}}{d_{[a]\cup[b]}}.
  \end{gather}
  Recall that all projectors $\Pi_i$ in these definitions are assumed to
  commute. The following relations hold:
  \begin{align*}
    &\text{(i).}
      &&\tfrac{1}{m}\meas_{a_1,\dots,a_k}\circ \meas_\Sigma(\id/d)
        \leq \frac{1}{k}\meas_{a_1,\dots,a_k}(\id/d)
      &&\text{for disjoint $\Pi_{a_i}$},\\
    &\text{(ii).}
      &&\chan_a\circ\meas_a(\id/d) = \tr[\Pi_a/d]\;\frac{\id}{d}.
        \phantom{\frac{1}{k}}\\
    &\text{(iii).}
      &&\meas_a = \meas_a\circ\meas_a,\\
    &\text{(iv).}
      &&\meas_{a}\circ\meas_b = \meas_{b}\circ\meas_a
        = \meas_{a,b}
      &&\text{for disjoint $\Pi_a,\Pi_b$},
        \phantom{\frac{1}{k}}\\
    &\text{(v).}
      &&\chan_a\circ\chan_b = \chan_b\circ\chan_a
        = \chan_{a,b}
      &&\text{for disjoint $\Pi_a,\Pi_b$},
        \phantom{\frac{1}{k}}\\
    &\text{(vi).}
      &&\meas_a\circ\chan_b = \chan_b\circ\meas_a
      &&\text{for disjoint $\Pi_a,\Pi_b$},
        \phantom{\frac{1}{k}}\\
    &\text{(vii).}
      &&\meas_a\circ\meas_\Sigma = \meas_\Sigma\circ\meas_a,
        \phantom{\frac{1}{k}}
  \end{align*}
\end{lemma}
(When we write $\meas\leq\mathcal{N}$ for CP maps $\meas,\mathcal{N}$, we mean
that $\meas(X)\leq\mathcal{N}(X)$ for any positive operator $X$.)

\begin{proof}
  To prove Part~(i), we start by applying almost exactly the same argument as
  in the proof of \cref{lem:halting_measurement_process}, replacing $\Pi_a$
  with $\bigotimes_{i=1}^k\Pi_{a_i}$ and noting that this tensor product of
  projectors kills $k$ terms from the sum in \cref{eq:kill_terms}, rather than
  just the one in the proof of \cref{lem:halting_measurement_process}. Thus
  \begin{equation}\label{eq:kill_k_terms}
  \begin{split}
    \meas_{a_1,\dots,a_k}\circ\meas_{\mathrm{cont}}^t(\id/d)
    &\leq \frac{1}{m}\sum_{t=0}^\infty \sum_{i_0,\dots,i_t \neq a} \frac{1}{m^t}
      \Biggl(\bigotimes_{i=1}^k\Pi_{a_i}\Biggr)
      \frac{\id}{d} \Biggl(\bigotimes_{i=1}^k\Pi_{a_i}\Biggr)\\
    &= \left(\frac{m-k}{m}\right)^t\meas_{a_1,\dots,a_k}(\id/d)
  \end{split}
  \end{equation}
  and
  \begin{subequations}
  \begin{align}
    \frac{1}{m}\meas_{a_1,\dots,a_k}\circ\meas_\Sigma(\id/d)
    &=\frac{1}{m}\sum_{t=0}^\infty\meas_{a_1,\dots,a_k}\circ
      \meas_{\mathrm{cont}}^t(\id/d)\\
    &\leq \frac{1}{m}\sum_{t=0}^\infty \left(\frac{m-k}{m}\right)^t
      \meas_{a_1,\dots,a_k}(\id/d)\\
    &= \frac{1}{k}\meas_{a_1,\dots,a_k}(\id/d).
  \end{align}
  \end{subequations}

  Part~(ii) follows immediately from the definitions of $\chan_a$ and
  $\meas_a$. Part~(iii) is immediate from $\Pi_a^2 = \Pi_a$. Parts~(iv)
  and~(v) follow trivially from the fact that two disjoint projectors by
  definition only act non-trivially on different subsystems. Similarly, in
  Part~(vi), since $\Pi_a$ and $\Pi_b$ are disjoint, the partial trace in
  $\chan_b$ is over a subsystem on which the projector $\Pi_a$ in $\meas_a$
  acts trivially. Finally, Part~(vii) follows from the fact that all the
  $\Pi_i$ commute.
\end{proof}

We are now in a position to prove \cref{lem:perpetual_measurement_process}.
\begin{proof}[of \cref{lem:perpetual_measurement_process}]
  The unnormalised density operator $X_{\Pi_{a_1},\Pi_{a_2},\dots,\Pi_{a_t}}$
  corresponding to the process of \cref{lem:perpetual_measurement_process}
  producing the sequence $\Pi_{a_1},\Pi_{a_2},\dots,\Pi_{a_t}$ is given by
  \begin{multline}\label{eq:seq_density_operator}
    X_{\Pi_{a_1},\Pi_{a_2},\dots,\Pi_{a_t}}\\\quad
    =\chan_{a_t}\circ\tfrac{1}{m}\meas_{a_t}\circ\meas_\Sigma\circ
      \chan_{a_{t-1}}\circ\tfrac{1}{m}\meas_{a_{t-1}}\circ\meas_\Sigma
      \circ\cdots\circ
      \chan_{a_1}\circ\tfrac{1}{m}\meas_{a_1}\circ\meas_\Sigma(\id/d).
  \end{multline}

  We start by using Part~(iii) of \cref{lem:CP_map_identities} to duplicate
  every $\meas_{a_i}$ which can be commuted all the way through the expression
  to right using Parts~(iv), (vi) and~(vii). This will produce a string of
  $\meas_{a_i}$'s of the form
  $\meas_{a_i}\circ\meas_{a_j}\circ\cdots\circ\meas_{a_1}$ next to the
  inner-most $\meas_{a_1}$. Since each $\meas_{a_j}$ is preceded by a
  $\chan_{a_j}$ in \cref{eq:seq_density_operator}, and $\meas_{a_i}$ can only
  be commuted through $\chan_{a_j}$ if $\Pi_{a_i}$ and $\Pi_{a_j}$ are
  disjoint, all $\meas_{a_i}$'s in this string necessarily have disjoint
  $\Pi_{a_i}$, so we can combine them into a single
  $\meas_{a_i,a_j,\dots,a_1}$ using Part~(iv) of \cref{lem:CP_map_identities}.

  At this point, the expression has the form
  \begin{multline}
   X_{\Pi_{a_1},\Pi_{a_2},\dots,\Pi_{a_t}}\\\quad
   =\chan_{a_t}\circ\tfrac{1}{m}\meas_{a_t}\circ\meas_\Sigma\circ
      \chan_{a_{t-1}}\circ\tfrac{1}{m}\meas_{a_{t-1}}\circ\meas_\Sigma
      \circ\cdots\circ
      \chan_{a_1}\circ\tfrac{1}{m}\meas_{a_i,a_j,\dots,a_1}\circ\meas_\Sigma(\id/d).
  \end{multline}
  We can now use Part~(i) of \cref{lem:CP_map_identities} to eliminate the
  inner-most $\meas_\Sigma$:
  \begin{multline}\label{eq:eliminate1}
    \chan_{a_t}\circ\tfrac{1}{m}\meas_{a_t}\circ\meas_\Sigma\circ
      \chan_{a_{t-1}}\circ\tfrac{1}{m}\meas_{a_{t-1}}\circ\meas_\Sigma
      \circ\cdots\circ
      \chan_{a_1}\circ\tfrac{1}{m}\meas_{a_i,a_j,\dots,a_1}\circ\meas_\Sigma(\id/d)\\
    \leq \chan_{a_t}\circ\tfrac{1}{m}\meas_{a_t}\circ\meas_\Sigma\circ
      \chan_{a_{t-1}}\circ\tfrac{1}{m}\meas_{a_{t-1}}\circ\meas_\Sigma
      \circ\cdots\circ
      \chan_{a_1}\circ\tfrac{1}{k}\meas_{a_i,a_j,\dots,a_1}(\id/d),
  \end{multline}
  where $k$ is the number of $\Pi_{a_i}$'s in the $\meas_{a_i,a_j,\dots,a_1}$.

  Now we reverse the above steps. First we separate out
  $\meas_{a_i,a_j,\dots,a_1}$ again into
  $\meas_{a_i}\circ\meas_{a_j}\circ\cdots\circ\meas_{a_1}$ using Part~(iv) of
  \cref{lem:CP_map_identities}. Then we use Parts~(iv), (vi) and~(vii) to move
  all the separated $\meas_{a_i}$ (except $\meas_{a_1}$) back through the
  expression to the left, until they can be recombined using Part~(iii) with
  the $\meas_{a_i}$ that they were originally duplicated from. We are left
  with
  \begin{multline}
    X_{\Pi_{a_1},\Pi_{a_2},\dots,\Pi_{a_t}}
    \leq \chan_{a_t}\circ\tfrac{1}{m}\meas_{a_t}\circ\meas_\Sigma\circ
      \chan_{a_{t-1}}\circ\tfrac{1}{m}\meas_{a_{t-1}}\circ\meas_\Sigma
      \circ\cdots\circ
      \chan_{a_1}\circ\tfrac{1}{k}\meas_{a_1}(\id/d).
  \end{multline}
  We now use Part~(ii) of \cref{lem:CP_map_identities} to arrive at
  \begin{multline}\label{eq:seq_density_operator_reduced}
    X_{\Pi_{a_1},\Pi_{a_2},\dots,\Pi_{a_t}}\\\quad
    \leq \frac{1}{k}\tr[\Pi_{a_i}/d]\cdot
      \chan_{a_t}\circ\tfrac{1}{m}\meas_{a_t}\circ\meas_\Sigma\circ
      \chan_{a_{t-1}}\circ\tfrac{1}{m}\meas_{a_{t-1}}\circ\meas_\Sigma
      \circ\cdots\circ
      \chan_{a_2}\circ\tfrac{1}{m}\meas_{a_2}\circ\meas_\Sigma(\id/d).
  \end{multline}
  Comparing with \cref{eq:seq_density_operator}, we see that, in this way, we
  have eliminated the inner-most
  $\chan_{a_1}\circ\tfrac{1}{m}\meas_{a_1}\circ\meas_\Sigma$ terms from the
  expression, picking up a scalar factor of $1/k$ in the process which depends
  on how many $\meas_{a_i}$ could be commuted through far enough to reach the
  inner-most~$\meas_{a_1}$.

  The sequence of CP maps on the right hand side of
  \cref{eq:seq_density_operator_reduced} now has exactly the same form as the
  original \cref{eq:seq_density_operator}, but shortened by omiting the $a_1$
  terms. Thus we have obtained a recursive operator inequality for
  $X_{\Pi_{a_1},\Pi_{a_2},\dots,\Pi_{a_t}}$:
  \begin{equation}\label{eq:X_recursion}
    X_{\Pi_{a_1},\Pi_{a_2},\dots,\Pi_{a_t}}
    \leq \frac{1}{k}\tr[\Pi_{a_i}/d]\cdot X_{\Pi_{a_2},\dots,\Pi_{a_t}}.
  \end{equation}
  Applying this inequality recursively, we finally arrive at
  \begin{equation}\label{eq:combinatorial_coefficient}
    X_{\Pi_{a_1},\Pi_{a_2},\dots,\Pi_{a_t}}
    \leq \frac{1}{\mathcal{P}}\cdot\prod_{i=1}^t\tr[\Pi_i\rho]\cdot
    \frac{\id}{d},
  \end{equation}
  with a yet-to-be-determined combinatorial coefficient $1/\mathcal{P}$
  arising from the $1/k$ coefficients introduced by Part~(i) of
  \cref{lem:CP_map_identities}.

  It remains to determine the combinatorial coefficient $1/\mathcal{P}$. This
  is a product of all the $1/k$ coefficients contributed by using Part~(i) of
  \cref{lem:CP_map_identities} to eliminate an $\meas_\Sigma$. Now, each
  $\meas_\Sigma$ is eliminated by a map of the form
  $\meas_{a_\alpha,a_\beta,\dots}$ (cf.\ \cref{eq:eliminate1}). Let
  $S=\{\Pi_{a_i},\Pi_{a_j},\dots\}$ be the set of elements from the original
  sequence $\Pi_{a_1},\Pi_{a_2},\dots,\Pi_{a_t}$ that are involved in the map
  $\meas_{a_\alpha,a_\beta,\dots}$ that eliminates the $\meas_\Sigma$ in
  question. Then the elimination of $\meas_\Sigma$ by
  $\meas_{a_\alpha,a_\beta,\dots}$ contributes a factor of $1/\abs{S}$. Thus,
  to understand the combinatorial coefficient $1/\mathcal{P}$, it suffices to
  understand all the sets $S$. It will be useful in what follows to consider
  both the subset $S$ of elements from the original sequence, and the set of
  vertices $V(S)$ from the resample DAG
  $\mathfrak{G}=\mathfrak{G}(\Pi_{a_1},\Pi_{a_2},\dots,\Pi_{a_t})$ that are
  labelled by the elements of $S$.

  Each $\meas_\Sigma$ appears immediately to the right of a particular
  $\meas_{a_i}$ in the original expression of \cref{eq:seq_density_operator},
  which in turn corresponds to a particular element $\Pi_{a_i}$ of the
  original sequence. Thus each $\meas_\Sigma$ can be identified with a unique
  element $\Pi_{a_i}$ of the sequence. Consider a particular $\meas_\Sigma$,
  and let $\Pi_{a_s}$ be the element of the original sequence with which it is
  identified. Now, the map $\meas_{a_\alpha,a_\beta,\dots}$ that eliminates
  this $\meas_\Sigma$ is made up of all those $\meas_{a_{i>s}}$ lying to the
  \emph{left} of $\meas_{a_s}$ in the original expression
  \cref{eq:seq_density_operator} which could be commuted past all the
  intervening maps using Parts~(iv), (vi) and~(vii) of
  \cref{lem:CP_map_identities}. Note that $\meas_{a_i}$ can only be moved past
  a $\chan_{a_j}$ if $\Pi_{a_i}$ is disjoint from $\Pi_{a_j}$. Thus we can
  determine the set $S$ corresponding to the $\meas_{a_\alpha,a_\beta,\dots}$
  that eliminates a given $\meas_\Sigma$ just by looking at the original
  sequence alone.

  $S$ can be found by the following procedure, which simply mirrors on the
  level of the sequence the steps involved in eliminating the $\meas_\Sigma$.
  (Cf.\ the reduction of \cref{eq:seq_density_operator} to
  \cref{eq:seq_density_operator_reduced}.) Start with the final element of the
  sequence,\footnote{Recall that the maps $\meas_{a_i}$ in the composition of
    \cref{eq:seq_density_operator} are in reverse order to the corresponding
    elements $\Pi_{a_i}$ in the sequence.} and move it as far as possible
  towards the beginning of the sequence, subject to the rule that it can only
  be moved past an element if it is disjoint from that element. Repeat this
  procedure once\footnote{Once we have finished moving a particular element,
    we are not allowed to move it again later, even if subsequent moves have
    ``unblocked'' it.} for each element $\Pi_{i>s}$ that appears later in the
  sequence than $\Pi_{a_s}$, starting from the final element $\Pi_{a_t}$
  working backwards through $\Pi_{a_{i>s}}$ until $\Pi_{a_s}$. $S$ is then the
  set of all elements that have been moved to the left of $\Pi_{a_s}$,
  including $\Pi_{a_s}$ itself (but not including any elements
  $\Pi_{a_{i<s}}$, which were never moved).

  How do the elements of $S$ relate to vertices in the resample DAG
  $\mathfrak{G}$? If any descendant of $\Pi_{a_i}$ appears in between
  $\Pi_{a_s}$ and $\Pi_{a_i}$ in the sequence, then $\Pi_{a_i}$ cannot be in
  $S$. To see this, note that by \cref{def:resample_DAG} all ancestors of an
  element must appear later in the sequence than that element.\footnote{Recall
    from \cref{def:resample_DAG} that the resample DAG is constructed by
    working \emph{backwards} through the sequence.} Therefore, if $\Pi_{a_j}$
  is a descendant of $\Pi_{a_i}$, then there is at least one child (possibly
  $\Pi_{a_j}$ itself) of $\Pi_{a_i}$ in between $\Pi_{a_s}$ and $\Pi_{a_i}$.
  But by \cref{def:resample_DAG} $\Pi_{a_i}$ intersects with all its children,
  so it cannot be moved past this child.

  Thus no element of $S$ has any descendant appearing between it and
  $\Pi_{a_s}$. But, as we've already noted, any descendant of an element must
  appear earlier in the sequence than that element. Therefore, no element of
  $S$ has any descendants appearing \emph{anywhere} after $\Pi_{a_s}$ in the
  sequence. Conversely, if none of the descendents of an element $\Pi_{a_i}$
  appear after $\Pi_{a_s}$ in the sequence, then $\Pi_{a_i}$ is in $S$. This
  is because the only thing that could stop $\Pi_{a_i}$ from being moved past
  $\Pi_{a_s}$ is a projector with which it intersects appearing between it and
  $\Pi_{a_s}$, but by \cref{def:resample_DAG} any such projector will be a
  descendant of $\Pi_{a_i}$.

  $S$ is therefore the set of all elements appearing after $\Pi_{a_s}$ in the
  sequence which have no descendants amongst the elements $\Pi_{a_{i>s}}$. In
  other words, if we consider the resample DAG
  $\mathfrak{G}_s(\Pi_{a_s},\Pi_{a_{s+1}},\dots,\Pi_{a_t})$ for the
  subsequence $\Pi_{a_s},\Pi_{a_{s+1}},\dots,\Pi_{a_t}$, then the elements of
  $S$ are the leaf vertices of $\mathfrak{G}_s$. Note that $\mathfrak{G}_s$
  can be obtained from the full resample DAG
  $\mathfrak{G}(\Pi_{a_s},\Pi_{a_{s+1}},\dots,\Pi_{a_t})$ by removing the
  vertices labelled by $\Pi_{a_1},\dots,\Pi_{a_{s-1}}$.

  Recall that the elimination of the $\meas_\Sigma$ corresponding to
  $\Pi_{a_s}$ contributes a factor $1/\abs{S}$ to the combinatorial
  coefficient $1/\mathcal{P}$, and we have seen that $S$ is the set of all
  leaves remaining in the resample DAG after vertices labelled by elements
  preceding $\Pi_{a_s}$ have been removed. So $1/\abs{S}$ is just the
  probability of removing $\Pi_{a_s}$ when a leaf is picked uniformly at
  random and removed, given that vertices labelled by elements preceding
  $\Pi_{a_s}$ have already been removed from the resample DAG. There is one
  such factor $1/\abs{S}$ contributed by each element $\Pi_{a_s}$ of the
  sequence, so the overall combinatorial coefficient $1/\mathcal{P}$ in
  \cref{eq:combinatorial_coefficient}, which is the product of all these
  $1/\abs{S}$ factors, is precisely the DAG probability of
  \cref{def:DAG_probability}. This completes the proof of
  \cref{eq:seq_density_operator_inequality}.

  The probability bound of \cref{eq:seq_probability_bound} follows immediately
  from \cref{eq:seq_density_operator_inequality} and the fact that the
  probability of the obtaining the sequence corresponding to the unnormalised
  density operator $X_{a_1,a_2,\dots,a_t}$ is given by
  $\tr[X_{a_1,a_2,\dots,a_t}]$.
\end{proof}

\Cref{lem:perpetual_measurement_process} immediately implies a simple bound on
the probability of the measurement outcomes from the
\cref{lem:perpetual_measurement_process} process forming a given resample DAG.
\begin{corollary}\label{cor:resample_DAG_probability}
  Let $G$ be a fixed resample DAG with vertices $v\in V(G)$ labelled by
  $\Pi_v$, and let $\Pi_{a_1},\Pi_{a_2},\dots,\Pi_{a_t}$ be the first $t$
  measurement outcomes produced by the
  \cref{lem:perpetual_measurement_process} process. Then
  \begin{equation}
    \Pr\left[\mathfrak{G}(\Pi_{a_1},\Pi_{a_2},\dots,\Pi_{a_t}) = G\right]
    \leq \prod_{v\in V(G)}\tr[\Pi_v\rho].
  \end{equation}
\end{corollary}
\begin{proof}
  To stand any chance of producing the DAG $G$, we must have
  $\{\Pi_{a_1},\dots,\Pi_{a_t}\} = \{\Pi_v\}_{v \in V(G)}$. The probability of
  producing a given resample DAG $G$ is obtained by summing the probabilities
  of all distinct sequences that generate that DAG. Let
  \begin{equation}
    \mathcal{S}=\{s=\Pi_{a_{\sigma(1)}},\Pi_{a_{\sigma(2)}},\dots,\Pi_{a_{\sigma(t)}}
                  \mid\mathcal{G}(s)=G\}
  \end{equation}
  be the set of all sequences with resample DAG $G$. Note that this is just
  the set of all permutations $\sigma$ of $\{\Pi_v\}$ consistent with the
  partial order encoded by $G$.

  Now, every sequence $s\in\mathcal{S}$ can be generated with probability
  $p_G(s)$ by running the process of \cref{def:DAG_probability} on $G$;
  conversely, every sequence generated by running that process on $G$ is in
  $\mathcal{S}$.
  Since the process always generates \emph{some} sequence, we have
  \begin{equation}
    \sum_{s\in\mathcal{S}}p_G(s) = 1.
  \end{equation}
  Thus, summing the probabilities from \cref{eq:seq_probability_bound} of
  \cref{lem:perpetual_measurement_process}, we have
  \begin{equation}
    \begin{split}
      \Pr\left[\mathfrak{G}(\Pi_{a_1},\Pi_{a_2},\dots,\Pi_{a_t}) = G\right]
      &= \sum_{\mathclap{s:\mathfrak{G}(s)=G}}\Pr(s)
      \leq \sum_{s\in\mathcal{S}}p_G(s) \cdot
         \prod_{i=1}^t\tr[\Pi_{a_i}\rho]\\
      &= \prod_{v\in G}\tr[\Pi_v\rho] \cdot \sum_{s\in\mathcal{S}}p_G(s)
      = \prod_{v\in G}\tr[\Pi_v\rho].
    \end{split}
  \end{equation}
\end{proof}

\subsection{Partial Resample DAGs}
Later on, we will need to consider a variant of the resample DAG, constructed
almost exactly as in \cref{def:resample_DAG}, but with one key difference.
Instead of including all projectors from the sequence in the DAG, we only
include those that are reachable from the vertex labelled by the final
projector in the sequence; i.e.\ any projector that cannot be attached to an
existing vertex when constructing the DAG (cf.\ \cref{def:resample_DAG}) is
simply discarded.

\begin{definition}[Partial Resample DAG]
  \label{def:partial_resample_DAG}
  Let $\{\Pi_i\}$ be a finite set of projectors acting on a tensor product
  space $\bigotimes_i\HS_i$, and let $\Pi_{a_1},\Pi_{a_2},\dots,\Pi_{a_l}$ be
  a sequence of projectors chosen from this set. The \keyword{partial resample
    DAG} is the subgraph of the resample DAG
  $\mathfrak{G}(\Pi_{a_1},\Pi_{a_2},\dots,\Pi_{a_l})$ which consisting of the
  subset of vertices that are reachable from the vertex labelled by
  $\Pi_{a_l}$, and all edges that begin and end at vertices in this subset.
\end{definition}

The \keyword{partial resample DAG}
$\mathfrak{g}(\Pi_{a_1},\Pi_{a_2},\dots,\Pi_{a_l})$ can be constructed from
the sequence as follows. Start with the final projector $\Pi_{a_l}$ in the
sequence. Create a root vertex and label it with $\Pi_{a_l}$. Then, working
\emph{backwards} through the sequence, for each projector $\Pi_{a_{i<l}}$ find
the set $L_i$ of all vertices labelled by projectors that intersect with
$\Pi_{a_i}$. If $L_i$ is empty, skip this projector. Otherwise, create a new
vertex, labelling it by $\Pi_{a_i}$, and attach it to the DAG by creating
directed edges from each element of $L_i$ to the new vertex.

\begin{definition}[Relevant Subsequence]
  \label{def:partial_DAG_subsequence}
  Let $\Pi_{a_1},\Pi_{a_2},\dots,\Pi_{a_t}$ be a sequence of projectors, and
  $\mathfrak{g}(\Pi_{a_i},\Pi_{a_2},\dots,\Pi_{a_t})$ the corresponding
  partial resample DAG. The \keyword{relevant subsequence} of
  $\Pi_{a_1},\Pi_{a_2},\dots,\Pi_{a_t}$ is the subsequence obtained by
  discarding from the sequence all elements that get discareded when
  $\mathfrak{g}(\Pi_{a_i},\Pi_{a_2},\dots,\Pi_{a_t})$ is constructed according
  to \cref{def:partial_DAG_probability}. Note that two sequences with the same
  relevant subsequence necessarily generate the same partial resample DAG.
\end{definition}

\begin{definition}[Partial DAG Probability]
  \label{def:partial_DAG_probability}
  Let $\{\Pi_i\}$ be a finite set of projectors acting on a tensor product
  space $\bigotimes_i\HS_i$, and let $\mathfrak{g}$ be a partial resample DAG
  over these projectors. The \keyword{partial DAG probability} with respect to
  $\mathfrak{g}$ of the sequence $\Pi_{a_1},\Pi_{a_2},\dots,\Pi_{a_l}$,
  denoted $p_\mathfrak{g}(\Pi_{a_1},\Pi_{a_2},\dots,\Pi_{a_l})$, is the
  probability that the following process generates the sequence. Starting with
  the partial resample DAG $\mathfrak{g}$ and an empty initial sequence, pick
  one of the leaf vertices uniformly at random and append its label to the end
  of the sequence, removing the vertex from the DAG. Repeat this procedure
  until no vertices are left.
\end{definition}

The analogous result to \cref{cor:resample_DAG_probability} also holds for
partial resample DAGs.
\begin{lemma}\label{lem:partial_resample_DAG}
  Let $g$ be a fixed partial resample DAG with vertices $v\in V(g)$
  labelled by projectors $\Pi_v$, and let
  $\Pi_{a_1},\Pi_{a_2},\dots,\Pi_{a_t}$ be the first length-$t$ relevant
  subsequence of outcomes from the \cref{lem:perpetual_measurement_process}
  process. Then
  \begin{equation}
    \Pr[\mathfrak{g}(\Pi_{a_1},\Pi_{a_2},\dots,\Pi_{a_t}) = g]
    \leq \prod_{v\in V(g)}\tr[\Pi_v\rho].
  \end{equation}
\end{lemma}
\begin{proofpart}
  The proof is similar to the proofs of
  \cref{lem:perpetual_measurement_process,cor:resample_DAG_probability},
  but some of the outcomes generated by the
  \cref{lem:perpetual_measurement_process} process might not be irrelevant
  because they get discarded when constructing a partial resample DAG. We
  will first bound the probability of a given sequence
  $\Pi_{a_1},\Pi_{a_2},\dots,\Pi_{a_t}$ forming the first length-$t$ relevant
  subsequence. We can then sum over all relevant subsequences that generate
  the specified partial resample DAG $g$ to bound the overall probability
  of producing $g$.

  Previously, the map $\meas_\Sigma$ allowed for any number of $\id-\Pi$
  measurement outcomes in \cref{lem:perpetual_measurement_process}. We must
  now also allow for any number of irrelevant outcomes. A projector is
  discarded if and only if it does not intersect with any projector that
  occurs \emph{later} in the sequence.\footnote{Once again, recall that the
    partial resample DAG is constructed \emph{backwards}.} So, between any
  two elements $\Pi_{a_{i-1}}$ and $\Pi_{a_i}$ of the sequence
  $\Pi_{a_1},\Pi_{a_2},\dots,\Pi_{a_t}$, we can have any number of outcomes
  $\Pi_x$ such that $\Pi_x$ is disjoint from all $\Pi_{a_{j>i}}$, as well as
  any number of $\id-\Pi$ outcomes.

  In other words, if $\mathcal{X}_i$ is the set of irrelevant outcomes that
  will be discarded from the DAG if they occur between elements
  $\Pi_{a_{i-1}}$ and $\Pi_{a_i}$ in the sequence, then the generalised
  measurement element corresponding to a single iteration of the
  \cref{lem:halting_measurement_process} sub-process \emph{not} producing a
  relevant outcome is given by
  \begin{equation}
    \begin{split}
      \meas_{\mathrm{cont}}^{(i)}(\sigma)
      &= \sum_j\frac{1}{m}(\id-\Pi_j)\sigma(\id-\Pi_j)
        + \sum_{\mathclap{\Pi_a\in\mathcal{X}_i}}\chan_a\circ
          \tfrac{1}{m}\meas_a(\sigma)\\
      &= \frac{1}{m}\sum_{\Pi_a\not\in\mathcal{X}_i}(\id-\Pi_a)\sigma(\id-\Pi_a)
        + \frac{1}{m}\sum_{\Pi_a\in\mathcal{X}_i}\mathcal{T}_a(\sigma)
    \end{split}
  \end{equation}
  where the trace-preserving CP map
  \begin{equation}
    \begin{split}
      \mathcal{T}_a(\sigma)
      &= (\id-\Pi_a)\sigma(\id-\Pi_a) + \chan_a\circ\meas_a\\
      &= (\id-\Pi_a)\sigma(\id-\Pi_a)
        + \rho_{[a]}\otimes\tr_{[a]}(\Pi_a\sigma\Pi_a).
    \end{split}
  \end{equation}
  The map corresponding to any number of irrelevant outcomes occurring is then
  given by
  \begin{equation}
    \meas_\Sigma^{(i)}(\sigma)
    = \sum_{t=0}^\infty(\meas_{\mathrm{cont}}^{(i)})^t(\sigma).
  \end{equation}
  We now prove a version of Part~(i) of \cref{lem:CP_map_identities} for
  the map $\meas_\Sigma^{(i)}$.
\end{proofpart}

\begin{lemma}\label{lem:another_CP_map_identity}
  Let $\Pi_{a_\alpha},\Pi_{a_\beta},\dots$ be disjoint from each other
  \emph{and} from all $\Pi_x\in\mathcal{X}_i$, and let $\tr_{[\mathcal{X}_i]}$
  denote the partial trace over all qudits on which some
  $\Pi_x\in\mathcal{X}_i$ acts
  non-trivially. Then
  \begin{equation}
    \tr_{[\mathcal{X}_i]}\left[
      \tfrac{1}{m}\meas_{a_\alpha,a_\beta,\dots}\circ
      (\meas_\Sigma^{(i)})^t(\id/d) \right]
    \leq\frac{1}{k}\tr_{[\mathcal{X}_i]}[\meas_{a_\alpha,a_\beta,\dots}(\id/d)].
  \end{equation}
\end{lemma}
\begin{proof}
  First, note that, since $\Pi_{a_1},\dots,\Pi_{a_k}$ are disjoint from each
  other and from $\Pi_x$, we have
  \begin{equation}\label{eq:Ta_identity}
    \begin{split}
      \tr_{[x]}\left[
        \meas_{a_\alpha,a_\beta,\dots}\circ\mathcal{T}_x(\sigma)
      \right]
      &=\tr_{[x]}\left[
        \mathcal{T}_x\circ\meas_{a_\alpha,a_\beta,\dots}(\sigma)
      \right]\\
      &=\tr_{[x]}\left[\meas_{a_\alpha,a_\beta,\dots}(\sigma)\right],
    \end{split}
  \end{equation}
  since $\mathcal{T}_x$ is trace-preserving and only acts non-trivially on
  $[X]$.

  Now,
  \begin{multline}\label{eq:TrX_induction}
    \tr_{[\mathcal{X}_i]}\left[
      \meas_{a_\alpha,a_\beta,\dots}\circ(\meas_{\mathrm{cont}}^{(i)})^t(\id/d)
    \right]\\
    =\tr_{[\mathcal{X}_i]}\Biggl[
       \meas_{a_\alpha,a_\beta,\dots}\biggl(
       \frac{1}{m}\sum_{\Pi_a\not\in\mathcal{X}_i}
       (\id-\Pi_a)(\meas_{\mathrm{cont}}^{(i)})^{t-1}(\id/d)(\id-\Pi_a)
     \biggr)\\
    +\frac{1}{m}\sum_{\Pi_a\in\mathcal{X}_i}
      \meas_{a_\alpha,a_\beta,\dots}\circ\mathcal{T}_a\circ
      (\meas_{\mathrm{cont}}^{(i)})^{t-1}(\id/d)
    \Biggl].
  \end{multline}
  Since $\Pi_{a_\alpha},\Pi_{a_\beta},\dots$ are disjoint from all
  $\Pi_a\in\mathcal{X}_i$, all the terms involving projectors that intersect
  with $\Pi_{a_\alpha},\Pi_{a_\beta},\dots$ occur in the first sum. The first
  sum can therefore be treated exactly as in \cref{eq:kill_k_terms} from
  \cref{lem:CP_map_identities}, giving
  \begin{multline}
    \tr_{[\mathcal{X}_i]}\Biggl[
      \meas_{a_\alpha,a_\beta,\dots}\biggl(
        \frac{1}{m}\sum_{\Pi_a\not\in\mathcal{X}_i}
        (\id-\Pi_a)(\meas_{\mathrm{cont}}^{(i)})^{t-1}(\id/d)(\id-\Pi_a)
      \biggr)
    \Biggr]\\
    \leq\tr_{[\mathcal{X}_i]}\Biggl[
      \frac{m-\abs{\mathcal{X}_i}-k}{m}\;\meas_{a_\alpha,a_\beta,\dots}
      \circ(\meas_{\mathrm{cont}}^{(i)})^{t-1}(\id/d)
    \Biggr].
  \end{multline}
  Using \cref{eq:Ta_identity}, the second sum simplifies to
  \begin{subequations}
  \begin{align}
    \tr_{[\mathcal{X}_i]}&\Biggl[
      \frac{1}{m}\sum_{\Pi_a\in\mathcal{X}_i}
      \meas_{a_\alpha,a_\beta,\dots}\circ\mathcal{T}_a\circ
      (\meas_{\mathrm{cont}}^{(i)})^{t-1}(\id/d)
    \Biggl]\notag\\
    &=\tr_{[\mathcal{X}_i]}\Biggl[
      \frac{1}{m}\sum_{\Pi_a\in\mathcal{X}_i}
      \meas_{a_\alpha,a_\beta,\dots}\circ
      (\meas_{\mathrm{cont}}^{(i)})^{t-1}(\id/d)
    \Biggl]\\
    &=\tr_{[\mathcal{X}_i]}\Biggl[
      \frac{\abs{\mathcal{X}_i}}{m}
      \meas_{a_\alpha,a_\beta,\dots}\circ
      (\meas_{\mathrm{cont}}^{(i)})^{t-1}(\id/d)
    \Biggl].
  \end{align}
  \end{subequations}
  Putting these together, we obtain
  \begin{equation}
    \tr_{[\mathcal{X}_i]}\left[
      \meas_{a_\alpha,a_\beta,\dots}\circ
      (\meas_{\mathrm{cont}}^{(i)})^t(\id/d)
    \right]
    \leq\tr_{[\mathcal{X}_i]}\Biggl[
      \frac{m-k}{m}\meas_{a_\alpha,a_\beta,\dots}\circ
      (\meas_{\mathrm{cont}}^{(i)})^{t-1}(\id/d)
    \Biggr].
  \end{equation}
  By induction on $t$ (the base case $t=0$ is trivial) and summing over $t$,
  we obtain the identity in the lemma.
\end{proof}

\begin{proof}[of \cref{lem:partial_resample_DAG}, continued]\hfill\newline
  The probability of obtaining $\Pi_{a_1},\Pi_{a_2},\dots,\Pi_{a_t}$ for the
  first $t$ relevant outcomes is given by
  \begin{multline}
    \Pr(\Pi_{a_1},\Pi_{a_2},\dots,\Pi_{a_t})\\
    =\tr\Bigl[
       \chan_{a_t}\circ\tfrac{1}{m}\meas_{a_t}\circ
       \meas_\Sigma^{(t)}\circ
       \chan_{a_{t-1}}\circ\tfrac{1}{m}\meas_{a_{t-1}}\circ
       \meas_\Sigma^{(t-1)}\circ\cdots\\
       \cdots\circ\chan_{a_1}\circ\tfrac{1}{m}\meas_{a_1}\circ
       \meas_\Sigma^{(1)}(\rho)
     \Bigr].
  \end{multline}
  (Compare with \cref{eq:seq_density_operator} from
  \cref{lem:perpetual_measurement_process}. The $\meas_\Sigma^{(i)}$ now also
  account for outcomes that are discarded when constructing the partial
  resample DAG.)

  We eliminate the $\meas_\Sigma^{(i)}$ from this expression in much the same
  way as in \cref{lem:perpetual_measurement_process}, using
  \cref{lem:another_CP_map_identity} instead of \cref{lem:CP_map_identities},
  Part~(i). Since a projector is discarded when constructing the DAG if and
  only if it is disjoint from \emph{all} projectors that occur \emph{later} in
  the sequence, they do not affect which $\meas_{a_i}$'s can be commuted all
  the way to the right, so iterating the same sequence of steps again leads to
  a recursive operator inequality of the form
  \begin{equation}
    X_{\Pi_{a_1},\Pi_{a_2},\dots,\Pi_{a_t}}
    \leq \frac{1}{k}\tr[\Pi_{a_i}/d]\cdot X_{\Pi_{a_2},\dots,\Pi_{a_t}}.
  \end{equation}
  The only difference as compared to the inequality in \cref{eq:X_recursion}
  lies in the combinatorial factor $1/k$, which now only counts relevant
  measurement outcomes.

  The combinatorial factors arising from \cref{lem:another_CP_map_identity}
  are otherwise identical to those of \cref{lem:CP_map_identities}, Part~(i),
  leading to an overall factor given this time by the \emph{partial} DAG
  probability $p_\mathfrak{g}(\Pi_{a_1},\Pi_{a_2},\dots,\Pi_{a_t})$ with
  $\mathfrak{g} = \mathfrak{g}(\Pi_{a_1},\Pi_{a_2},\dots,\Pi_{a_t})$. Thus we
  finally end up with the following bound on the probability of obtaining
  $\Pi_{a_1},\Pi_{a_2},\dots,\Pi_{a_t}$ as the first $t$ outcomes that are not
  discarded from the partial resample DAG:
  \begin{equation}
    \Pr(\Pi_{a_1},\Pi_{a_2},\dots,\Pi_{a_t})
    \leq p_\mathfrak{g}(\Pi_{a_1},\Pi_{a_2},\dots,\Pi_{a_t})\cdot
    \prod_{i=1}^t\tr[\Pi_{a_i}\rho].
  \end{equation}

  As in \cref{cor:resample_DAG_probability}, when we sum over all relevant
  subsequences that generate the specified partial resample DAG $g$, the
  combinatorial coefficients $p_g$ sum to 1 and we arrive at the bound claimed
  in the lemma.
\end{proof}

\subsection{Expected number of violations}
We are now in a position to re-prove the bound of
\cref{thm:expected_violations} on the expected number of violations seen by
\cref{alg:quantum}. Recall that \cref{alg:quantum} keeps a log of violated
measurements $\Pi_{a_1},\Pi_{a_2},\dots,\Pi_{a_t}$, whose only purpose is to
aid in the analysis.
We say that a given partial resample DAG $g$ \keyword{occurs} in
\cref{alg:quantum}'s log if it can be constructed starting from some entry of
the log, i.e.\ if for some $l$ $\mathfrak{g}(\Pi_{a_1},\dots,\Pi_{a_l}) = g$.

The process implemented by \cref{alg:quantum} is just the iterated measurement
process defined in \cref{lem:perpetual_measurement_process}. So
\cref{lem:partial_resample_DAG} immediately implies the following result:
\begin{corollary}\label{cor:partial_DAG_probability}
  The probability that the partial resample DAG $g$ occurs in the log is
  at most $\prod_{v\in V(g)} R(\Pi_v)$, where $V(g)$ is the vertex set of $g$,
  $\Pi_v$ is the projector labelling vertex $v$, and the relative dimension
  $R(\Pi_v)$ is equivalent to the probability of measuring $\Pi_v$ on the
  maximally mixed state.
\end{corollary}
(Though proved by very different means, \cref{cor:partial_DAG_probability} is
essentially a quantum version of Lemma~2.1 from \citet{MoserTardos}, though
that Lemma is not quite a special case of \cref{cor:partial_DAG_probability}
as it concerns ``witness trees'' rather than partial resample DAGs.)

In order to bound the expected number of violations seen by
\cref{alg:quantum}, we again follow \citet{MoserTardos} in relating
\cref{cor:partial_DAG_probability} to a Galton-Watson branching process.
Notice that the bounds in
\cref{cor:resample_DAG_probability,cor:partial_DAG_probability} are
determined solely by the set of vertex labels appearing in the DAG; the
structure of the DAG plays no role. Let $\tau(g)$ be a (labelled) spanning
tree for the partial resample DAG $g$. Because it is the spanning tree for
an resample DAG, all the children of a vertex in $\tau(g)$ carry distinct
labels (cf.\ \cref{def:partial_resample_DAG}). Labelled trees with this
property will be called \keyword{proper}, irrespective of whether they are
spanning trees for some DAG; all spanning trees are proper, but not all proper
trees are spanning trees.

We will once again relate the bound in \cref{cor:partial_DAG_probability} to
the probability that the Galton-Watson process of \cref{lem:Galton-Watson}
generates a proper tree $\tau_a$ whose root vertex is labelled by some fixed
projector $\Pi_a$, yielding an alternative proof of
\cref{thm:expected_violations}.
%
\begin{proof}[of \cref{thm:expected_violations}]
  Let $N_a$ be the number of times a given projector $\Pi_a$ appears in
  \cref{alg:quantum}'s log, which is the same as the number of times $\Pi_a$
  was measured to be violated. Let $\mathcal{G}_a$ denote the set of all
  partial resample DAGs whose root vertex is labelled by $\Pi_a$, and let
  $\mathcal{T}_a$ denote the set of all proper trees with root vertex labelled
  by $\Pi_a$. Then, from \cref{cor:partial_DAG_probability}, we have
  \begin{equation}
    \begin{split}
      \expectation(N_a)
      &=\sum_{g_a\in\mathcal{G}_a}\Pr(\text{$g_a$ occurs in the log})
      \leq \sum_{g_a\in\mathcal{G}_a}\quad\prod_{\mathclap{v\in V(g_a)}} R(\Pi_v)\\
      &=\sum_{g_a\in\mathcal{G}_a}\quad\;\;\prod_{\mathclap{v\in V(\tau(g_a))}} R(\Pi_v)
      \leq \sum_{\tau_a\in\mathcal{T}_a}\quad\prod_{\mathclap{v\in V(\tau_a)}} R(\Pi_v).
    \end{split}
  \end{equation}
  The final relation holds because distinct resample DAGs have distinct
  spanning trees. It is an inequality for two reasons: a DAG can have multiple
  spanning trees, leading to double-counting, and the set of proper trees is a
  strict superset of the set of spanning trees.

  Now, by the assumption of the QLLL (\cref{thm:constructive_QLLL}), the
  relative dimension satisfies $R(\Pi_i) \leq x'_i$. Thus
  \begin{equation}
    \expectation(N_a)
    \leq \sum_{\tau_a\in\mathcal{T}_a}\quad\prod_{\mathclap{v\in V(\tau_a)}} R(\Pi_v)
    \leq \sum_{\tau_a\in\mathcal{T}_a}\quad\prod_{\mathclap{v\in V(\tau_a)}} x'_v
    \leq \sum_{\tau_a\in\mathcal{T}_a} \frac{x_a}{1-x_a}\Pr(\tau_a),
  \end{equation}
  the final inequality following from \cref{lem:Galton-Watson}, $\Pr(\tau_a)$
  being the probability of the Galton-Watson process generating tree $\tau_a$.
  Since that process either produces a proper tree in $\mathcal{T}_a$, or
  continues indefinitely, we have $\sum_{\tau_a\in\mathcal{T}_a}\Pr(\tau_a)
  \leq 1$, thus
  \begin{equation}
    \expectation(N_a)
    \leq \frac{x_a}{1-x_a}\sum_{\tau_a\in\mathcal{T}_a} \Pr(\tau_a)
    \leq \frac{x_a}{1-x_a},
  \end{equation}
  and the theorem follows from summing over all projectors.
\end{proof}

Using this alternate proof of \cref{thm:expected_violations}, the rest of the
proof of the constructive commutative QLLL (\cref{thm:constructive_QLLL}) goes
through as before using the arguments of
\cref{sec:eff_QLLL,sec:constructive_QLLL_proof}.

\section{Application: Bounding Convergence Times of CP Maps}
\label{sec:CP_map_convergence}
A (time-homogeneous) \keyword{quantum Markov process} is generated by
iterating a fixed completely positive trace-preserving (CP) map $\chan$, so
that $t$ steps of the process are described by the composition
\begin{equation}
  \chan^t = \underbrace{\chan\circ\chan\circ{\dotsm}\circ\chan}_t.
\end{equation}
As in the case of (classical) Markov chains, convergence times of quantum
stochastic processes are an important field of study, with applications to
quantum dynamics \citep{AAKV00,PWC09,CEW09}, quantum information theory
\citep{PWC09}, and quantum algorithms \citep{VWC09,TOVPF09,STV11}.

The results of \cref{sec:constructive_proof} prove that
\cref{alg:QLLL_converger} converges in polynomial time to a state satisfying
the requirements of the commutating QLLL, or equivalently the ground state
subspace of the associated Hamiltonian. This holds promise as a new technique
for proving fast convergence of quantum Markov processes to their steady
state. However, strictly speaking \cref{alg:QLLL_converger} does not implement
a quantum Markov process. Although each iteration of \cref{alg:quantum}, on
which it is based, does implement a fixed CP map $\chan$, the averaging trick
used in \cref{alg:QLLL_converger} means that the latter algorithm does not
simply iterate this map. (The overall evolution of \cref{alg:QLLL_converger}
for a given total time $t$ can of course be described by some CP map, but this
map is not of the form $\chan^t$.) \Cref{alg:exact_QLLL} is further still from
being a quantum Markov process.

In this section, we instead focus on the quantum Markov process defined by
\cref{alg:quantum}, and prove that if the set of projectors defining the
process satisfy the Lov\'asz conditions of \cref{thm:constructive_QLLL}, then
this process converges quickly to the steady-state subspace (in time
polynomial in the overlap with the steady-state subspace, and polynomial in
the spectral gap of the associated Hamiltonian).

Proving fast convergence is slightly more involved than in the case of
\cref{alg:QLLL_converger}, as we cannot use the averaging trick of
\cref{sec:constructive_QLLL_proof}, but instead must combine the fact that the
overlap with the steady-state subspace is monotonic under the CP map, with
Poissonian accumulation of the probability of measuring a violation, leading
to a contradiction with \cref{thm:expected_violations} were the convergence
rate too slow.

\begin{theorem}\label{thm:CP_map_convergence}
  For a given set of commuting projectors $\Pi_1,\Pi_2,\dots,\Pi_m$, let
  $\chan$ be the CP map defined by\footnote{This is the same as the map shown
    in \citet{VWC09} to converge to the ground state of any frustration-free
    Hamiltonian, but with no control on the convergence rate.}
  \begin{equation}\label{eq:QLLL_map}
    \chan(\rho) = \frac{1}{m}\sum_i (\id-\Pi_i)\rho(\id-\Pi_i)
      + \tr_{\!/[i]}[\Pi_i\rho]\otimes\frac{\id_{[i]}}{d^k}.
  \end{equation}
  If the projectors satisfy the Lov\'asz conditions of
  \cref{thm:constructive_QLLL}, then the Markov process produced by iterating
  this map converges in $\order{m/\delta^2\varepsilon}$ iterations to a state
  $\rho$ with fidelity $\tr[P_0\rho] \geq 1-\varepsilon$ with the steady-state
  subspace of $\chan$ (where $P_0$ is the projector onto that subspace, and
  $\delta$ is the spectral gap of $H=\frac{1}{m}\sum_i\Pi_i$).
\end{theorem}
Note, that according to the definition above the spectral gap $\delta=1/m$ in
the case of $m$ commuting projectors considered here.
\begin{proof}
  Note that \cref{eq:QLLL_map} is precisely the map implemented by one
  iteration of \cref{alg:quantum}. If the projectors satisfy the Lov\'asz
  conditions, then we know from \cref{thm:QLLL} that $P_0$ is non-zero. Let
  $\rho_t = \chan^t(\rho_0)$ be the state of \cref{alg:quantum}'s assignment
  register at time $t$. We can decompose $\rho_t$ in terms of $P_0$:
  \begin{equation}\label{eq:rho_decomposition}
    \rho_t = P_0\rho_tP_0 + P_0\rho_t(\id-P_0)
             + (\id-P_0)\rho_tP_0 + (\id-P_0)\rho_t(\id-P_0).
  \end{equation}
  Note that $P_0\rho_tP_0$ is an (unnormalised) state supported on $P_0$, the
  cross-terms are traceless, and $(\id-P_0)\rho_t(\id-P_0)$ is an
  (unnormalised) state supported on $\id-P_0$.

  Since $P_0\rho_tP_0$ is supported on the subspace of states that do not
  violate any projector, it is invariant under $\chan$, thus
  \begin{subequations} \label{eq:nondec}
  \begin{align}
    \begin{split}
      \tr[P_0\chan(\rho_t)]
      =&\tr[P_0\rho_t] + \tr[P_0\chan(P_0\rho_t(\id-P_0))]
        + \tr[P_0\chan((\id-P_0)\rho_tP_0)]\\
        &\qquad + \tr[P_0\chan((\id-P_0)\rho_t(\id-P_0))]
    \end{split}\raisetag{1.2em}\\
    &\geq \tr[P_0\rho_t],
  \end{align}
  \end{subequations}
  where the inequality follows from dropping non-negative terms.
  $\tr[P_0\rho_t]$ is therefore monotonically non-decreasing under $\chan$.

  Let $\alpha \coloneqq \tr[P_0\rho_t]$. The probability of a randomly chosen
  projector being violated by $\rho_t$ is given by
  $\frac{1}{m}\sum_i\tr[\Pi_i\rho_t]$. Since $\Pi_iP_0 = 0$ by definition, we
  have
  \begin{subequations}\label{eq:violation_upper_bound}
    \begin{align}
      \frac{1}{m}\sum_i\tr[\Pi_i\rho_t]
      &=\frac{1}{m}\sum_i\Bigl(
         \tr[\Pi_iP_0\rho_tP_0] + \tr[\Pi_iP_0\rho_t(\id-P_0)]\\
             &\mspace{85mu} + \tr[\Pi_i(\id-P_0)\rho_tP_0]
             + \tr[\Pi_i(\id-P_0)\rho_t(\id-P_0)]\Bigr)\notag\\
      &=\frac{1}{m}\sum_i\tr[\Pi_i(\id-P_0)\rho_t(\id-P_0)]\\
      &\leq \tr[(\id-P_0)\rho_t]\\
      &= 1 - \alpha,
    \end{align}
  \end{subequations}
  where the inequality follows from the fact that $\Pi_i \leq \id$.

  Recall that the spectral gap $\delta$ is defined as
  \begin{equation}
    \delta = \min_{\rho\;:\;P_0\rho=0}\; \frac{1}{m}\sum_i\tr[\Pi_i\rho].
  \end{equation}
  Note that $\delta > 0$, since $\delta = 0$ would imply that $\tr[\Pi_i\psi]
  = 0$ for all $\psi,\Pi_i$ which is impossible except in the trivial case of
  all $\Pi_i = 0$. Then
  \begin{subequations}\label{eq:violation_lower_bound}
    \begin{align}
      \frac{1}{m}\sum_i\tr[\Pi_i\rho_t]
      &=\frac{1}{m}\sum_i\Bigl(
         \tr[\Pi_iP_0\rho_t] + \tr[\Pi_iP_0\rho_t(\id-P_0)]\\
             &\mspace{85mu} + \tr[\Pi_i(\id-P_0)\rho_tP_0]
             + \tr[\Pi_i(\id-P_0)\rho_t(\id-P_0)]\Bigr)\notag\\
      &=\tr[(\id-P_0)\rho_t]\frac{1}{m}\sum_i\tr\left[
          \Pi_i\frac{(\id-P_0)\rho_t(\id-P_0)}{\tr[(\id-P_0)\rho_t]}
        \right]\\
      &\geq (1 - \alpha)\;\delta.
    \end{align}
  \end{subequations}
  Since $\alpha=\tr[P_0\rho_t]$ is non-decreasing with $t$, this also holds
  for all $\rho_{\tau<t}$.

  So, the probability of a randomly chosen projector being violated by
  $\rho_t$ is upper-bounded by $1-\alpha$, and the probability of a randomly
  chosen projector being violated by $\rho_{\tau<t}$ is lower-bounded by
  $(1-\alpha)\delta$. Imagine that we iterate the process for $t$ iterations.
  Given that $M$ violations occur, what is the probability that one of the
  violations occurs in the final ($t$\textsuperscript{th}) iteration? From
  \cref{eq:violation_upper_bound,eq:violation_lower_bound}, the worst-case
  distribution of violations is given by
  \begin{subequations}\label{eq:violation_distribution}
  \begin{gather}
    \prob(\text{violation in }t^{\mathrm{th}}\text{ iteration})
      \propto 1-\alpha,\\
    \prob(\text{violation in }\tau<t^{\mathrm{th}}\text{ iteration})
      \propto (1-\alpha)\;\delta.
  \end{gather}
  \end{subequations}

  The probability of a violation in the final iteration is then equivalent to
  the probability of picking the $t$\textsuperscript{th} element, when we pick
  $M$ from $t$ elements according to the distribution of
  \cref{eq:violation_distribution}. By the union bound, this is at most
  $M\times\prob({\text{pick }t^\mathrm{th}}\text{ element})$. Normalising the
  distribution of \cref{eq:violation_distribution}, we have
  \begin{gather}
    \alpha = 1 - \frac{1}{t\delta + 1},
  \end{gather}
  so
  \begin{equation}\label{eq:prob_final_violation}
    \prob(\text{violation in }t^{\mathrm{th}}\text{ iteration}
          \mid M\text{ violations})
      \leq \frac{M}{t\delta + 1}.
  \end{equation}

  We have a bound on the expected number of violations from
  \cref{thm:expected_violations}. Using this together with
  \cref{eq:prob_final_violation}, we have
  \begin{equation}\label{eq:P_violation}
    \prob(\text{violation in }t^{\mathrm{th}}\text{ iteration})
      \leq \sum_M\frac{M}{t\delta + 1} \Pr(M\text{ violations})
    =\frac{\expectation(M)}{t\delta + 1} \leq \frac{m}{t\delta+1}.
  \end{equation}
  From \cref{eq:violation_lower_bound}, we know that $\prob(\text{violation in
    $t^{\mathrm{th}}$ iteration}) \geq (1-\alpha)\;\delta$. Together with
  \cref{eq:P_violation}, this implies
  \begin{equation}
    \tr[P_0\rho_t] \eqqcolon \alpha
    \geq 1 - \frac{m}{t\delta(\delta + 1)}
    = 1 - \Order{\frac{m}{t\delta^2}}.
  \end{equation}
  Therefore, in order to achieve an overlap $\tr[P_0\rho_t] \geq
  1-\varepsilon$ with the steady-state subspace, we require $t =
  \order{m/\delta^2\varepsilon}$ iterations.
\end{proof}

Since the map $\chan$ of \cref{thm:CP_map_convergence} implements exactly the
dissipative state engineering map of \citet{VWC09},
\cref{thm:CP_map_convergence} implies the following result about the
convergence of this map.
\begin{corollary}\label{cor:ground-state_convergence}
  Consider the many-body Hamiltonian $H = \sum_{i=1}^m h_i$ on $n$ particles,
  where each local term $h_i$ acts on some subset of the particles, and all
  the $h_i$ mutually commute. Let $\Pi_i$ be the projector onto the orthogonal
  complement of the lowest-energy eigenspace of $h_i$.

  If the projectors satisfy the Lov\'asz conditions of
  \cref{thm:constructive_QLLL}, then the CP map $\chan$ defined in
  \cref{thm:CP_map_convergence} will converge in time
  $\order{m/\delta^2\varepsilon}$ to a state with fidelity $1-\varepsilon$
  with the ground state subspace of $H$, where $\delta$ is the spectral gap of
  the Hamiltonian $H_\Pi = \frac{1}{m}\sum_i \Pi_i$ in which the local terms
  are replaced by projectors.
\end{corollary}
Note that the convergence time is \emph{independent} of the number of
particles and of their local Hilbert space dimension.

\section{The Non-Commutative Case}\label{sec:non-commuting}
In the previous sections, we have overcome the first challenge to proving a
constructive QLLL: that of efficiently constructing the desired entangled
state using only local measurements. To prove a constructive version of the
general QLLL of~\citet{QLLL}, we must overcome the second challenge:
non-commutativity of the projectors, and the concomitant
measurement-disturbance problem.

In the non-commutative setting, there is a new parameter that can play a role
in the run-time of the QLLL algorithm: the spectral gap of the Hamiltonian
associated with a QLLL instance.
\begin{definition}[QLLL spectral gap]
  \label{def:spectral-gap}
  The \keyword{spectral gap} $\delta$ of a set of projectors $\{\Pi_i\}$ is
  the spectral gap of the associated Hamiltonian\footnote{Note our choice of
    normalisation, which follows that of \citet*{QLLL}, but differs from some
    papers in the literature by a factor of $m$.}
  \begin{equation}
    H = \frac{1}{m}\sum_i\Pi_i,
  \end{equation}
  i.e.\ the difference between the smallest and second-smallest eigenvalues
  (ignoring degeneracies).

  If the projectors satisfy the Lov\'asz conditions
  (\cref{def:Lovasz_conditions}), then the minimum eigenvalue of $H$ is~0 and
  the spectral gap can be expressed as
  \begin{equation}\label{eq:spectral_gap}
    \delta = \min_{\rho\;:\;P_0\rho=0}\; \frac{1}{m}\sum_i\tr[\Pi_i\rho],
  \end{equation}
  where $P_0$ is the projector on the subspace of states fulfilling the QLLL
  requirements, i.e.\ the projector onto $\vspan\{\ket{\psi}:\forall i\,
  \Pi_i\ket{\psi}=0\}$.
\end{definition}
Note that \cref{eq:spectral_gap} can be interpreted as the probability of
measuring a randomly chosen projector to be violated, minimised over all
states in the orthogonal complement of the subspace of states that do not
violate any projector.

\subsection{Witness trees}
The proof of the key \cref{lem:tau-check} for the commutative case in
\cref{sec:witness-trees}, leading to the bound in
\cref{thm:expected_violations} on the expected number of violations seen by
\cref{alg:quantum}, made crucial use of commutativity of the $\{\Pi_i\}$.
(Commutativity is also necessary for the alternative combinatorial proof given
in \cref{sec:combinatorial_proof}.) In principle, \cref{lem:tau-check} and
hence \cref{thm:expected_violations} could still hold in the non-commutative
case, without necessarily implying the same run-time as in the commutative
case. We will see in \cref{sec:converging} that, depending on the type of
convergence we demand, the spectral gap enters naturally in the required
run-time in the non-commutative setting, through using a bound on the expected
number of violations to prove a bound the rate of convergence to the desired
state.

However, there are strong indications that the bound in \cref{lem:tau-check}
\cref{part:prob} on the probability of a witness tree $\tau$ occurring in
\cref{alg:quantum}'s log no longer holds for non-commutative projectors. The
proof of \cref{lem:tau-check} still goes through in the commutative setting
if, instead of constructing a witness tree by discarding any projectors that
do not intersect with one already in the tree, we instead include these
projectors, thereby constructing a DAG instead of a tree. This still encodes
the partial ordering of the violations with respect to projector
intersections, so the coupling argument in \cref{sec:witness-trees} still
works, giving the analogous bound on the probability of these DAGs occurring
in the log (still in the commutative case). Indeed, we proved exactly this
result for resample DAGs in \Cref{cor:resample_DAG_probability}, by more
involved linear algebra and combinatorial techniques.

But the following counter-example shows that this bound on the probability of
a given DAG occurring in the log \emph{no longer holds} if the projectors do
not commute.
\begin{example}\label{ex:violation}
  Consider a system of two qubits. Let $\Pi_1 = \proj{0} \otimes \id$, $\Pi_2
  = \id \otimes \proj{0}$, and $\Pi_3 = \proj{\psi} + \proj{01} + \proj{10}$,
  where $\ket{\psi}=\sqrt{a}\ket{00}+\sqrt{b}\ket{11}$. Let $\tau$ be the
  graph consisting of two (unconnected) vertices labelled by $\Pi_1$ and
  $\Pi_2$.

  If \cref{cor:resample_DAG_probability} or (the generalisation of)
  \cref{lem:tau-check} were true in the non-commutative case, the probability
  of $\tau$ occurring in \cref{alg:quantum}'s log would be bounded by
  \begin{equation}\label{eq:violated_bound}
    \Pr(\tau) \leq \tr(\Pi_1/d) \tr(\Pi_2/d) = 1/4.
  \end{equation}

  There are two ways $\tau$ can occur: either the sequence $\Pi_1,\Pi_2$
  occured at the beginning of the log, or the same violations occurred in the
  other order $\Pi_2,\Pi_1$. The probability of these two sequences is the
  same, by symmetry. Direct analytical calculation gives
  \begin{equation}
    \Pr(\tau) = 2\Pr(\Pi_1,\Pi_2) =
    \begin{cases}
      \frac{1}{9} + \frac{7a}{24(1+a)} + \frac{b(11+12a)}{144(1+a)^2}
        & \text{if } a < 1 \\
      \frac{1}{9} & \text{if } a=1,
    \end{cases}
  \end{equation}
  which violates the bound in \cref{eq:violated_bound} for $a >
  \frac{3}{20}(\sqrt{41}-1)$, with the probability tending to $\frac{37}{144}$
  as $a \rightarrow 1$ (a violation of $\frac{1}{144}$).
\end{example}

\subsubsection{A Conjecture}
Nevertheless, we conjecture that a similar bound does still hold in general,
weakened by a polynomial factor.
\begin{conjecture}\label{conj:cubitt-schwarz_poly}
  Let $\{\Pi_i\}$ be an arbitrary set of projectors, let $\tau$ be a fixed
  witness tree with vertices labelled by these projectors, and $L$ the log
  produced by running \cref{alg:quantum} with these projectors. We conjecture
  that the probability that $\tau$ occurs in $L$ is at most
  \begin{equation}
    \Pr(\tau) \leq \poly(\Abs{\tau}) \prod_{v\in \tau}\prob[\Pi(v)],
  \end{equation}
  where $\Abs{\tau}$ is the size of $\tau$ (total number of vertices).
\end{conjecture}
(We implicitly also allow here a polynomial dependence on the other parameters
$n$, $m$, $d$, $x_i$ and $1/\delta$. Only the dependence on $\tau$ plays a
role in the subsequent proof; any dependence on the other parameters carries
straight through to the final run-time.)

In \cref{sec:non-commuting_expected_violations}, we will see that this
conjecture relates the probability of a tree occurring in the general quantum
case to the expectation value of a functional over the \keyword{total progeny}
of the Galton-Watson process described in \cref{sec:expected_violations}. In
the classical and commutative cases, we saw that this conjecture held for the
trivial constant functional $\poly(\abs{\tau})=1$.
%

For the simplest case of trees containing a single vertex, we can prove
\cref{conj:cubitt-schwarz_poly} even \emph{without} the polynomial factor
(i.e.\ with $\poly(\Abs{\tau}) = 1$). The single-vertex case is not completely
trivial, as an arbitrary number of satisfied measurements can occur before
\cref{alg:quantum} sees its first violation. In fact, we will prove a slightly
stronger \emph{operator} bound, which immediately implies the conjecture for
single-vertex trees. We already proved this result for commuting projectors in
\cref{lem:halting_measurement_process}. The following proposition generalises
that result to arbitrary sets of projectors.
\begin{proposition}\label{prop:first_violation}
  Let $\{\Pi_i\}$ be a set of $m$ projectors on $\CC^d$. Consider the
  following iterated measurement process, starting from the maximally mixed
  state. In each step, a projector $\Pi_i$ is chosen independently uniformly
  at random, and the two-outcome measurement $\{\Pi_i,\id-\Pi_i\}$ is
  performed. This is repeated until a $\Pi_i$ outcome is obtained, at which
  point the process halts.

  Let $\rho_a$ denote the final state of the system given that it halted on
  outcome $\Pi_a$, and $p_a$ the probability that this occurs. Then the
  (unnormalised) density matrix corresponding to the process halting on
  outcome $\Pi_a$ satisfies
  \begin{subequations}
  \begin{equation}
    X_a = p_a\rho_a \leq \Pi_a/d,
  \end{equation}
  hence
  \begin{equation}
    p_a = \tr(X_a) \leq \tr(\Pi_a)/d.
  \end{equation}
  \end{subequations}
\end{proposition}
This result is quite striking. It tells us that we if start from the maximally
mixed state and perform binary, two-outcome projective measurements at random
until we obtain the first ``0'' outcome, we can effectively ignore all
intermediate measurements withboutcome ``1'', and replace this process with
that of simply measuring the final outcome on the maximally mixed state
directly.

\begin{proof}
  We will in fact prove a slightly more general result. Let $A^{-1}$ denote
  the Moore-Penrose pseudoinverse of a matrix $A$, and $A^*$ denote entry-wise
  complex conjugation. The operation $vec(A)$ or $\ket{A}$ denotes
  vectorisation of the matrix $A$, i.e.\ the treatment of matrix
  $A\in\mathcal{M}_d(\CC)$ as a vector in $\mathbbm{C}^{d^2}$. By a slight
  abuse of notation, we sometimes write $\ket{X} \leq \ket{Y}$ to mean $X \leq
  Y$.

  We can express the unnormalised density matrix $X_a$ as
  \begin{equation}
    X_a = \frac{1}{m} \Pi_a \sum_{t=0}^{\infty} \left(
      \sum_{i_1,\dots,i_t}\frac{1}{m}(\id-\Pi_{i_t}) \cdots \left(
        \sum_{i_1} \frac{1}{m}(\id-\Pi_{i_1}) \frac{\id}{d} (\id-\Pi_{i_1})
      \right)
      \cdots (\id-\Pi_{i_t}) \right) \Pi_a,
  \end{equation}
  or, vectorising,
  \begin{equation}
    \ket{X_a} = \frac{1}{m} [\Pi_a^*\ox\Pi_a] \;
      \sum_{t=0}^\infty \left(
        \frac{1}{m} \sum_i (\id-\Pi_i^*)\ox(\id-\Pi_i)
      \right)^t
      \frac{1}{d} \ket{\id}.
  \end{equation}
  But this is a Neumann series $\sum_{t=0}^{\infty} T^t = (\id - T)^{-1}$.
  Thus the infinite sum can be written as
  \begin{align}
    \ket{X_a}
    &=\frac{1}{m} [\Pi_j^*\ox\Pi_j] \left(
        \id - \frac{1}{m} \sum_i (\id-\Pi_i^*)\ox(\id-\Pi_i)
      \right)^{-1}
      \frac{1}{d} \ket{\id}\\
    &=(\Pi_j^*\ox\Pi_j) \left(
        \sum_i (\id\ox\Pi_i) + (\Pi_i^*\ox\id) - (\Pi_i^*\ox\Pi_i)
      \right)^{-1}
      \frac{1}{d}\ket{\id}.
  \end{align}

  For a more compact notation, define $Q = \sum_i \Pi_i$,
  $B_i=\Pi_i^*\ox\Pi_i$, $B = \sum_i B_i$, and $A = \sum_i (\id\ox\Pi_i +
  \Pi_i^*\ox\id) = Q^*\ox\id + \id\ox Q$. Then the claim of the
  Proposition is equivalent to
  \begin{equation}
    B_i (A-B)^{-1} \ket{\id} \leq B_i \ket{\id}.
  \end{equation}

  We will prove a slightly more general result. Let
  $P=\Pi_{i_1}\Pi_{i_2}\cdots \Pi_{i_k}$ be a product of $k$ \emph{commuting}
  projectors chosen from $\{\Pi_i\}$, and define $B_P = P^*\ox P$. We claim
  that
  \begin{equation} \label{eq:shortclaim}
    B_P (A-B)^{-1} \ket{\id}
    \leq \left(\frac{1}{2}+\frac{1}{2k}\right) B_P \ket{\id},
  \end{equation}
  the Proposition being equivalent to the $k=1$ case of this. To show this, we
  need the following lemmas, which we prove below.

  \begin{lemma}\label{lem:tian}
    $P Q^{-1} P \leq \frac{1}{k} P$.
  \end{lemma}

  \begin{lemma}\label{lem:A}
    $A^{-1}\ket{\id} = \frac{1}{2}\ket{Q^{-1}} + \ket{Q_\perp}$, where
    $Q_\perp$ is an operator whose support (coimage) lies in the kernel of
    $Q$.
  \end{lemma}

  \begin{lemma}\label{lem:BA}
    $B_P A^{-1} \ket{\id} \leq \frac{1}{2k} B_P\ket{\id}$, and in particular
    $B A^{-1} \ket{\id} \leq \frac{1}{2} B\ket{\id}$.
  \end{lemma}

  \begin{lemma}\label{lem:BAB}
    $B_P A^{-1} B \ket{\id} = \frac{1}{2} B_P\ket{\id}$, and in particular $B
    A^{-1} B\ket{\id} = \frac{1}{2} B\ket{\id}$.
  \end{lemma}

  \begin{lemma}\label{lem:BAt}
    Let $t \geq 1$ be an integer. Then $(BA^{-1})^t\ket{\id} \leq
    \frac{1}{2^t} B\ket{\id}$.
  \end{lemma}
  \cref{eq:shortclaim} now follows from
  \begin{subequations}
  \begin{align}
    B_P (A-B)^{-1} \ket{\id}
    &= B_P(A(\id-A^{-1}B))^{-1}\ket{\id} \label{sprt}
    = B_P(\id-A^{-1}B)^{-1}A^{-1}\ket{\id}\\
    &= B_P\sum_{t=0}^\infty (A^{-1}B)^tA^{-1}\ket{\id}
    = B_P A^{-1}\sum_{t=0}^\infty (BA^{-1})^{t}\ket{\id}\\
    &= B_P A^{-1}\ket{\id} + B_P A^{-1}\sum_{t=1}^{\infty}(BA^{-1})^t\ket{\id}\\
    &\leq B_P A^{-1}\ket{\id}
      + B_P A^{-1}\sum_{t=1}^{\infty}\frac{1}{2^t}B\ket{\id}
      \label{a1}\\
    &= B_P A^{-1}\ket{\id} + B_P A^{-1} B \ket{\id}
    \leq \frac{1}{2k} B_P \ket{\id} + \frac{1}{2} B_P \ket{\id}\label{a2}\\
    &= \left(\frac{1}{2}+\frac{1}{2k}\right) B_P \ket{\id},
  \end{align}
  \end{subequations}
  where the first equality in \cref{sprt} follows as $B$ is in the support of
  $A$, \cref{a1} from \cref{lem:BAt}, and the inequality in \cref{a2} from
  \cref{lem:BA} and \cref{lem:BAB}.
\end{proof}

We now prove the five lemmas used above.
\begin{proof}[of \cref{lem:tian}]
  First note that
  \begin{equation}
    \sum_i \Pi_i \geq k\,\Pi_{i_1}\Pi_{i_2}\cdots \Pi_{i_k}
  \end{equation}
  because
  \begin{align}
    \sum_i \Pi_i - k\,\Pi_{i_1}\Pi_{i_2}\cdots \Pi_{i_k}
    &= \Pi_{i_1}+\Pi_{i_2}+\Pi_{i_k}+\sum_{\mathclap{j\notin \{i_1,...,i_k\}}} \Pi_j
    - k \Pi_{i_1}\Pi_{i_2}\cdots \Pi_{i_k}\\
    &\geq k \Pi_{i_1}\Pi_{i_2}\cdots \Pi_{i_k}
      + \sum_{\mathclap{j\notin \{i_1,...,i_k\}}} \Pi_j
      - k \Pi_{i_1}\Pi_{i_2}\cdots \Pi_{i_k}
      \label{prbnd}\\
    &= \sum_{\mathrlap{j\notin \{i_1,...,i_k\}}} \Pi_j \;
    \geq 0,
  \end{align}
  where \cref{prbnd} follows from lower bounding each term $\Pi_{i_n}$ by
  $\Pi_{i_n} P$, which is true for any projector. (Here, $P =
  \prod_{n=1}^k\Pi_{i_n}$ is a projector since the $\Pi_{i_n}$ commute by
  assumption.) But this implies
  \begin{equation}
    \Bigl(\sum_i \Pi_i\Bigr)^{-1}
    \leq \frac{1}{k}\Pi_{i_1}\Pi_{i_2}\cdots \Pi_{i_k} = \frac{1}{k} P,
  \end{equation}
  and the lemma follows by left- and right-multiplying by $P$.
\end{proof}

\begin{proof}[of \cref{lem:A}]
  We consider the equation $A^{-1}\ket{\id} = \ket{\sigma}$ and solve for the
  unknown $\ket{\sigma}$. By first multiplying with $A$ from the left, we see
  that $AA^{-1}\ket{\id}=A\ket{\sigma}$, where $AA^{-1}$ is the projector onto
  the support of $A$. Furthermore, note that $AA^{-1} = (QQ^{-1})^{*} \ox
  QQ^{-1}$, thus $AA^{-1}\ket{\id} = \ket{QQ^{-1}}$. Therefore, written in
  unvectorised form, we have to solve the matrix equation $Q\sigma + \sigma Q
  = QQ^{-1}$. Clearly the restriction of $\sigma$ to the support of $Q$ can
  only be $\frac{1}{2}Q^{-1}$, but we can add an arbitrary operator in the
  kernel of $Q$.
\end{proof}

\begin{proof}[of \cref{lem:BA}]
  By \cref{lem:A}, $B_P A^{-1}\ket{\id} = \frac{1}{2} B_P \ket{Q^{-1}}
  \eqqcolon \ket{\sigma}$, using the fact that $P$ is in the support of $Q$.
  Thus, unvectorised, we have $\frac{1}{2} P Q^{-1} P = \sigma$. But by
  \cref{lem:tian}, $\frac{1}{2}P Q^{-1} P \leq \frac{1}{2k} P$, thus
  $\ket{\sigma} \leq \frac{1}{2k}\ket{P} = \frac{1}{2k}B_P\ket{\id}$. The
  second part of the lemma follows simply by summing the $k=1$ case over $i$.
\end{proof}

\begin{proof}[of \cref{lem:BAB}]
  First, note that $B\ket{\id}=\ket{Q}$, thus $B_P A^{-1}B\ket{\id} = B_P
  A^{-1}\ket{Q}$. Defining $A^{-1}\ket{Q} \eqqcolon \ket{\sigma}$, we can
  multiply with $A$ from the left and solve $AA^{-1}\ket{Q}=A\ket{\sigma}$ for
  $\ket{\sigma}$. Now note that $AA^{-1}\ket{Q} = [(QQ^{-1})^{*}\otimes
  QQ^{-1}]\ket{Q} = \ket{Q}$. Unvectorising, we have $Q\sigma + \sigma Q = Q$,
  and clearly $\sigma=\frac{1}{2}\id + Q_\perp$, where we can add an arbitrary
  operator $Q_\perp$ in the kernel of $Q$. The lemma follows easily when we
  recall that $P$ is in the support of $Q$.
\end{proof}

\begin{proof}[of \cref{lem:BAt}]
  First we bound
  \begin{equation}
    (BA^{-1})^t\ket{\id}
      = (BA^{-1})^{t-1}BA^{-1}\ket{\id}
      \leq \frac{1}{2}(BA^{-1})^{t-1}B\ket{\id}
  \end{equation}
  using \cref{lem:BA}. To show
  $(BA^{-1})^{t-1}B\ket{\id}=\frac{1}{2^{t-1}}B\ket{\id}$ by induction, it
  suffices to show
  \begin{equation}
    BA^{-1}B\ket{\id} = \frac{1}{2}B\ket{\id},
  \end{equation}
  which is then iterated $t-1$ times. But this is already established by
  \cref{lem:BAB}.
\end{proof}

\subsubsection{A Weaker Conjecture}
\Cref{conj:cubitt-schwarz_poly} weakens the bound we proved in the commutative
case by a polynomial factor. In fact, we will see that we can even cope with
an \emph{exponential} factor, and still prove a constructive QLLL with
polynomial run-time, albeit at the expense of having to strengthen the
Lov\'asz conditions in the non-commutative case. (We recover the original
Lov\'asz conditions if the projectors commute.)

\begin{conjecture}\label{conj:cubitt-schwarz_exp}
  Let $\{\Pi_i\}$ be an arbitrary set of projectors, let $\tau$ be a fixed
  tree with vertices labelled by these projectors, and $L$ the log produced by
  running \cref{alg:quantum} with these projectors. We conjecture that the
  probability that $\tau$ occurs in $L$ is at most
  \begin{equation}
    \Pr(\tau) \leq \frac{1\phantom{^{\Abs{\tau}}}}{(m\delta)^{\Abs{\tau}}}
                   \prod_{v\in \tau}\prob[\Pi(v)],
  \end{equation}
  where $\Abs{\tau}$ is the size of $\tau$ (total number of vertices).
\end{conjecture}

The motivation for this conjecture comes from the following result, which once
again proves the conjecture for the simplest case of trees containing a single
vertex.
\begin{proposition}\label{prop:first_violation2}
  Let $\{\Pi_i\}$ be a set of $m$ projectors on $\CC^d$, and define $\delta$
  to be the spectral gap of the associated Hamiltonian
  \begin{equation}
    H = \frac{1}{m}\sum_i\Pi_i.
  \end{equation}
  Consider the following iterated measurement process, starting from the
  maximally mixed state. In each step, a projector $\Pi_i$ is chosen
  independently uniformly at random, and the two-outcome measurement
  $\{\Pi_i,\id-\Pi_i\}$ is performed. This is repeated until a $\Pi_i$ outcome
  is obtained, at which point the process halts.

  Let $\rho_a$ denote the final state of the system given that it halted on
  outcome $\Pi_a$, and $p_a$ the probability that this occurs. Then the
  (unnormalised) density matrix corresponding to the process halting on
  outcome $\Pi_a$ satisfies
  \begin{subequations}
  \begin{equation}
    X_a = p_a\rho_a \leq \frac{1}{m\delta}\Pi_a/d,
  \end{equation}
  hence
  \begin{equation}
    p_a = \tr(X_a) \leq \frac{1}{m\delta}\frac{\tr(\Pi_a)}{d}.
  \end{equation}
  \end{subequations}
\end{proposition}

For commuting projectors, we always have $\delta \geq 1/m$ (and equality holds
unless \emph{every} state outside the kernel of $H$ violates more than one
projector), and we recover \cref{lem:halting_measurement_process}. For
non-commuting projectors, $\delta$ can range between 0 and 1, but we would
expect the more difficult cases to be when $\delta < 1/m$, in which case
\cref{prop:first_violation2} is a weaker statement than that already proven in
\cref{prop:first_violation}.

Nonetheless, \cref{prop:first_violation2} is substantially easier to prove
than \cref{prop:first_violation}, which could suggest that
\cref{conj:cubitt-schwarz_exp} is the ``correct'' non-commutative
generalisation of the bound in \cref{lem:tau-check}, and
\cref{prop:first_violation} is a red-herring that only applies to the very
special case of single-vertex trees.

\begin{proof}
  Define for the linear map $\meas(\rho) \coloneqq
  \frac{1}{m}\sum_i(\id-\Pi_i)\rho(\id-\Pi_i)$, and let $\meas^t$ denote the
  $t$-fold composition of $\meas$. Then the unnormalised density matrix $X_a$
  is given by
  \begin{subequations}
  \begin{align}
    X_a &= \frac{1}{m} \Pi_a \sum_{t=0}^{\infty} \left(
      \sum_{i_1,\dots,i_t}\frac{1}{m}(\id-\Pi_{i_t}) \cdots \left(
        \sum_{i_1} \frac{1}{m}(\id-\Pi_{i_1}) \frac{\id}{d} (\id-\Pi_{i_1})
      \right)
      \cdots (\id-\Pi_{i_t}) \right) \Pi_a\\
    &=\frac{1}{m} \Pi_a \left(\sum_{t=0}^\infty \meas^t(\id/d) \right) \Pi_a.
      \label{eq:X_meas}
  \end{align}
  \end{subequations}

  Let $G$ be the projector onto the ground state subspace (kernel) of $H$, and
  $E = G^\perp$ be the projector onto the subspace of excited states (the
  support, or coimage, of $H$). Note that $\forall i: \Pi_i G = G\Pi_i = 0$
  and $\Pi_i E = E\Pi_i = \Pi_i$. For any scalars $\alpha,\beta$,
  \begin{subequations}
  \begin{align}
    \meas(\alpha G + \beta E)
    &= \frac{1}{m} \sum_{i=1}^m (\id-\Pi_i)(\alpha G + \beta E)(\id-\Pi_i)
    = \alpha G + \beta\frac{1}{m}\sum_{i=1}^m (E-\Pi_i)\\
    &= \alpha G + \beta E - H \leq \alpha G + \beta (1-\delta) E.
  \end{align}
  \end{subequations}
  Applying this operator inequality $t$ times starting from $\alpha = \beta =
  1$ gives
  \begin{equation}
    \meas^t(\id) \leq G + (1-\delta)^t E.
  \end{equation}
  Using this in \cref{eq:X_meas} leads to
  \begin{equation}
    X_a \leq \frac{1}{m}\sum_{t=0}^\infty
           \frac{1}{d} \Pi_a \left(G + (1-\delta)^t E \right) \Pi_a
        = \frac{1}{m}\sum_{t=0}^\infty \frac{1}{d}(1-\delta)^t \Pi_a
        = \frac{1}{m\delta}\Pi_a/d,
  \end{equation}
  which proves the proposition.
\end{proof}

\subsection{Expected number of violations}
\label{sec:non-commuting_expected_violations}
From the conjectures, we can derive bounds on the expected number of
violations seen by \cref{alg:quantum}.

\begin{definition}[$\epsilon$-strengthened Lov\'asz conditions]
  \label{def:strengthened_Lovasz_conditions}
  Let $\Pi_1,\Pi_2,\dots,\Pi_m$ be projectors that act on arbitrary subsets of
  $n$ qudits. We say that the set of projectors $\{\Pi_i\}$ satisfies the
  \keyword{$\epsilon$-strengthened Lov\'asz conditions} if there exist values
  $0 \leq x_1,x_2,\dots,x_m \leq 1$ such that
  \begin{equation}
    R(\Pi_i) \leq (1-\epsilon)x_i\cdot\prod_{\mathclap{\Pi_j\in \Gamma(\Pi_i)}} (1-x_j).
  \end{equation}
\end{definition}

\begin{theorem}\label{thm:expected_violations_poly}
  Let $\{\Pi_i\}$ be a set of $m$ projectors satisfying the
  $\epsilon$-strengthened Lov\'asz conditions of
  \cref{def:strengthened_Lovasz_conditions} for some $\epsilon > 0$. If
  \cref{conj:cubitt-schwarz_poly} is true, then there exists a constant $C$
  (which depends on $\epsilon$) such that the expected number of violations
  seen by \cref{alg:quantum} is bounded by
  \begin{equation}
    \expectation(\text{\textup{total number of violations}})
    \leq C\sum_{i=1}^m \frac{x_i}{1-x_i}.
  \end{equation}
\end{theorem}
\begin{proof}
  The proof is very similar to that of \cref{thm:expected_violations} in
  \cref{sec:expected_violations} for the commutative case. Let $N_a$ be the
  number of times a given projector $\Pi_a$ appears in \cref{alg:quantum}'s
  log, which is the same as the number of times $\Pi_a$ is measured to be
  violated. Let $\mathcal{T}_a$ denote all proper witness trees whose root
  vertex is labelled by $\Pi_a$. Then, using \cref{conj:cubitt-schwarz_poly},
  we have
  \begin{equation}
    \expectation(N_a)
    = \sum_{\tau\in\mathcal{T}_a}\Pr(\text{$\tau$ appears in the log})
    \leq \sum_{\tau\in\mathcal{T}_a}\poly(\Abs{\tau})\prod_{v\in\tau}\Pr[\Pi(v)].
  \end{equation}
  The probability $\Pr[\Pi_i]$ of a projector being violated on a random state
  is just given by its relative dimension $\Pr(\Pi_i) = R(\Pi_i)$. Since by
  assumption the projectors satisfy the Lov\'asz conditions of
  \cref{def:strengthened_Lovasz_conditions}, we have
  \begin{subequations}
  \begin{align}
    \expectation(N_a)
    &\leq \sum_{\tau\in\mathcal{T}_a}\poly(\Abs{\tau})\prod_{\Pi_i\in\tau} R(\Pi_i)
    \leq \sum_{\tau\in\mathcal{T}_a}
      \poly(\Abs{\tau})\prod_{\Pi_i\in\tau} (1-\epsilon) x'_i\\
    &\leq \frac{x_a}{1-x_a}\sum_{\tau\in\mathcal{T}_a}
      \poly(\Abs{\tau})(1-\epsilon)^{\Abs{\tau}}\Pr(\tau),
      \label{eq:poly_exp}
  \end{align}
  \end{subequations}
  where $\Pr(\tau_a)$ is the probability that the Galton-Watson process
  defined in \cref{lem:Galton-Watson} generates tree $\tau_a$. The inequality
  in \cref{eq:poly_exp} follows from \cref{lem:Galton-Watson}.

  The summand in \cref{eq:poly_exp} is just the expectation of a functional
  $f(x) = \poly(x)(1-\epsilon)^x$ over the \keyword{total progeny}
  $\Abs{\tau}$ of the Galton-Watson process. For our purposes, it sufficies to
  observe that this functional is always bounded. Let $C \coloneqq \max_{x\geq
    0} f(x)$. Then
  \begin{equation}
    \expectation(N_a)
    \leq \frac{x_a}{1-x_a}\sum_{\tau\in\mathcal{T}_a} C\Pr(\tau)
    \leq C\frac{x_a}{1-x_a},
  \end{equation}
  the final inequality coming from the fact that the Galton-Watson process
  either produces a tree in $\mathcal{T}_a$, or continues indefinitely. The
  theorem follows from summing over all projectors.
\end{proof}

In fact, to prove \cref{thm:expected_violations_poly}, all we require is that
the expectation of the functional $\poly(\abs{\tau})$ of the total progeny
$\tau$ of the Galton-Watson multi-type branching process be bounded. It is
therefore likely that the $1-\epsilon$ factor can be removed by a more careful
analysis of the distribution of total progeny for the branching process.

\begin{theorem}\label{thm:expected_violations_exp}
  Let $\{\Pi_i\}$ be a set of $m$ projectors with spectral gap $\delta$, which
  satisfy the $\epsilon$-strengthened Lov\'asz conditions for $\epsilon = 1 -
  m\delta$. If \cref{conj:cubitt-schwarz_exp} is true, then the expected
  number of violations seen by \cref{alg:quantum} is bounded by
  \begin{equation}
    \expectation(\text{\textup{total number of violations}})
    \leq \sum_{i=1}^m \frac{x_i}{1-x_i}.
  \end{equation}
\end{theorem}
\begin{proof}
  As in \cref{thm:expected_violations_poly},
  \begin{equation}
    \expectation(N_a)
    = \sum_{\tau\in\mathcal{T}_a}\Pr(\text{$\tau$ appears in the log})
      \leq \sum_{\tau\in\mathcal{T}_a}
      \frac{1}{(m\delta)^{\Abs{\tau}}}\prod_{v\in\tau}\Pr[\Pi(v)],
  \end{equation}
  this time using \cref{conj:cubitt-schwarz_exp}. Thus
  \begin{equation}
    \expectation(N_a)
    \leq \frac{x_a}{1-x_a}\sum_{\tau\in\mathcal{T}_a}
      \frac{(1-\epsilon)^{\Abs{\tau}}}{(m\delta)^{\Abs{\tau}}} \Pr(\tau)
    = \frac{x_a}{1-x_a}\sum_{\tau\in\mathcal{T}_a} \Pr(\tau)
    \leq \frac{x_a}{1-x_a},
  \end{equation}
  and the theorem follows by summing over all projectors.
\end{proof}

\subsection{Converging to a solution}\label{sec:converging}
In the non-commutative setting, there is an ambiguity in what it means for an
algorithm to construct a state satisfying the requirements of the QLLL of
\cref{thm:QLLL}. The non-constructive QLLL asserts the existence of a state
that does not violate any of the projectors. We could demand that the
algorithm converges to a state whose probability of violating any projector is
at most $\varepsilon$; we will call this ``weak convergence'', or
``convergence in energy''. Alternatively, we could demand that the algorithm
converges to a state that is $\varepsilon$-close (in some suitable distance
measure, say fidelity) to the subspace of states that do not violate any
projector; we will call this ``strong convergence'', or ``convergence in
fidelity''. In the classical and in the commutative cases, these are
equivalent, and there is a single, unambiguous definition of convergence.
However, non-commutativity of the projectors means that these two notions of
convergence are in general no longer equivalent in the quantum case.

\begin{definition}[Strong convergence, or convergence in fidelity]
  \label{def:strong-convergence}
  We say that an algorithm converges in the strong sense to a state fulfilling
  the QLLL if, given any $\varepsilon > 0$, there exists a $t$ such that the
  state $\rho$ produced by running the algorithm for time $t$ satisfies
  $\tr[P_0\rho] \geq 1-\varepsilon$, where $P_0$ is the projector onto the
  subspace of states defined by $\vspan\{\ket{\psi}:\forall i\,
    \Pi_i\ket{\psi}=0\}$; i.e.\ the state $\rho$ is $\varepsilon$-close in
  fidelity to the subspace of states fulfilling the QLLL requirements.
\end{definition}

\begin{definition}[Weak convergence, or convergence in energy]
  \label{def:weak-convergence}
  We say that an algorithm converges in the weak sense to a state fulfilling
  the QLLL if, given any $\varepsilon > 0$, there exists a $t$ such that the
  state $\rho$ produced by running the algorithm for time $t$ satisfies
  $\forall i\,\tr[\Pi_i\rho] \leq \varepsilon$.
\end{definition}

As our terminology suggests, strong convergence implies weak convergence, but
in general the converse does \emph{not} hold.
\begin{lemma}\label{lem:strong-stronger-than-weak}
  Strong convergence is strictly stronger than weak convergence.
\end{lemma}
\begin{proof}
  It is trivial to see that strong convergence implies weak convergence: if
  $\rho$ is such that $\tr[P_0\rho] \geq 1-\varepsilon/m$, then
  \begin{equation}
    \sum_i \tr[\Pi_i\rho]
    = \sum_i \tr[\Pi_i(\id-P_0)\rho]
    \leq \sum_i \tr[(\id-P_0)\rho] \leq \varepsilon
  \end{equation}
  (the first equality following from the fact that $\Pi_i P_0 = 0$ by
  definition, the middle inequality from the operator inequality
  $\Pi_i\leq\id$).

  To see that weak convergence is strictly weaker, consider the example of two
  almost-identical projectors acting on the same set of qudits:
  $\Pi_0=\proj{\psi}$ and $\Pi_1=\proj{\psi'}$, with
  $\abs{\braket{\psi}{\psi'}}^2 = 1-\delta$. Let
  \begin{equation}
    \ket{\psi^\perp} = \frac{(\id-\proj{\psi})\ket{\psi'}}
	                    {\sqrt{1-\delta}}
  \end{equation}
  be the state in the span of $\{\ket{\psi},\ket{\psi'}\}$ orthogonal to
  $\ket{\psi}$. For $\rho = \proj{\psi^\perp}$, we have
  \begin{gather}
    \tr[\Pi_0\rho] = \abs{\braket{\psi}{\psi^\perp}}^2 = 0,\\
    \tr[\Pi_1\rho]
      = \abs{\braket{\psi'}{\psi^\perp}}^2
      = \abs{\braKet{\psi'}{(\id-\proj{\psi})}{\psi'}}^2
      = \frac{\delta^2}{1-\delta},
  \end{gather}
  whereas $\tr[P_0\rho] = 0$. Therefore, letting $\delta$ tend to 0, an
  algorithm which converges to the state $\rho$ has converged to arbitrarily
  high precision in the weak sense, yet it is as far as possible from
  convergence in the strong sense.
\end{proof}

However, weak convergence \emph{does} imply strong convergence if we are
prepared to pay a price in the convergence time. That price is a dependence on
the spectral gap of \cref{def:spectral-gap}.
\begin{lemma}\label{lem:weak-strong-delta}
  If an algorithm converges weakly in time $\order{f(1/\varepsilon)}$, then it
  converges strongly in time $\order{f(1/m\delta\varepsilon)}$.
\end{lemma}
\begin{proof}
  If we run the algorithm for time $\order{f(1/m\delta\varepsilon)}$ then, by
  \cref{def:weak-convergence} of weak convergence, the state $\rho$ produced
  by the algorithm must satisfy $\forall i: \tr[\Pi_i\rho] \leq
  m\delta\varepsilon$. Decomposing the state $\rho$ as
  \begin{equation}
    \rho = P_0\rho P_0 + P_0\rho(\id-P_0)
            + (\id-P_0)\rho P_0 + (\id-P_0)\rho(\id-P_0),
  \end{equation}
  and noting that $\Pi_iP_0 = 0$ by definition, we have
  \begin{subequations}\label{eq:orthogonal_part_upper_bound}
    \begin{align}
      \begin{split}
        \delta\varepsilon
        &\geq \frac{1}{m}\sum_i\tr[\Pi_i\rho]\\
        &= \frac{1}{m}\sum_i\Bigl(
          \tr[\Pi_iP_0\rho P_0] + \tr[\Pi_iP_0\rho(\id-P_0)]\\
        &\phantom{= \frac{1}{m}\sum_i\Bigl(}
          + \tr[\Pi_i(\id-P_0)\rho P_0] + \tr[\Pi_i(\id-P_0)\rho(\id-P_0)]
        \Bigr)
      \end{split}\raisetag{5.15em}\\
      &= \frac{1}{m}\sum_i\tr[\Pi_i(\id-P_0)\rho(\id-P_0)].
    \end{align}
  \end{subequations}
  But, by \cref{def:spectral-gap} of the spectral gap $\delta$,
  \begin{equation}\label{eq:orthogonal_part_lower_bound}
    \frac{1}{m}\sum_i\frac{\tr[\Pi_i(\id-P_0)\rho(\id-P_0)]}{\tr[(\id-P_0)\rho]}
    \geq \delta,
  \end{equation}
  thus from
  \cref{eq:orthogonal_part_upper_bound,eq:orthogonal_part_lower_bound} we
  obtain $\tr[(\id-P_0)\rho] \leq \varepsilon$, or
  \begin{equation}
    \tr[P_0\rho] \geq 1-\varepsilon,
  \end{equation}
  proving strong convergence in time $\order{f(1/m\delta\varepsilon)}$ as
  claimed.
\end{proof}

Which is the ``correct'' notion of convergence depends on what one is trying
to achieve. From a physics perspective, an algorithm for finding a state
fulfilling the Lov\'asz conditions can be interpreted as an algorithm for finding
the ground state of the Hamiltonian
\begin{equation}\label{eq:QLLL_Hamiltonian}
  H = \frac{1}{m}\sum_i\Pi_i.
\end{equation}
If we are interested in cooling a system to the ground state, producing a
state with low energy is sufficient and weak convergence is the appropriate
notion, since the parameter $\varepsilon$ in \cref{def:weak-convergence}
upper-bounds the energy of the state with respect to the Hamilton $H$. On the
other hand, if we are interested in ground state properties, we want the
algorithm to produce a state close in fidelity to the true ground state,
requiring strong convergence.

If the spectral gap scales inverse-polynomially, then the run-time remains
polynomial even for strong convergence. For some important cases of the QLLL,
the gap does indeed scale inverse-polynomially, and the run-time is provably
polynomial even for strong convergence. (For example, in the commutative case
the gap is constant and we have already seen that an efficient algorithm
always exists.) On the other hand, there are certainly cases for which the
spectral gap is exponentially small, and in that case the run-time we require
to ensure strong convergence will be exponentially large. However, an
inverse-linear dependence on the gap may be the best one could hope to
achieve, at least for algorithms that work by measuring the projectors, as it
takes expected time $1/\delta$ to see even \emph{one} violation if the system
is in the lowesr excited state (a state necessarily \emph{orthogonal} to the
ground-state subspace).

From a complexity theoretic perspective, weak convergence is the more natural
notion, as it corresponds to the canonical FQMA-complete problem of low-energy
state preparation~\citep{FBQP}. (FQMA is the quantum analogue of FNP, the
functional version of NP; whereas an NP or QMA problem asks whether there
\emph{exists} an input producing a ``yes'' answer, an FNP or FQMA problem
requires such an input to be prepared.) Similarly, the \kQSAT{} problem, as
defined and shown to be QMA$_0$-complete by \citet{QSAT}, promises that a
given $k$-local Hamiltonian either has a ground state (minimum eigenvalue
eigenstate) with \emph{exactly} zero energy, or has a ground state with energy
larger than $\alpha$ (where $\alpha$ is inverse polynomial in the problem
size). The \kQSAT{} problem is then to decide whether a zero-energy state
exists. The corresponding FQMA problem is to produce a state that witnesses
the existence of a zero-energy state. Setting $\varepsilon \leq \alpha$ and
running an algorithm until it weakly converges to precision $\varepsilon$
produces precisely such a witness.

On the other hand, formulating strong convergence as a well-defined
complexity-theoretic concept requires a little care. For all the standard
complexity classes, the projectors defining a problem instance must be
specified classically. Thus if the problem size is to be well-defined, the
projectors can only be specified to finite precision. However, there exist
instances which are arbitrarily close to each other (the matrix elements of
the two sets of projectors are arbitrarily close), yet a state satisfying the
strong convergence requirements for one instance is maximally far from
satisfying those conditions for the other instance. (The two-projector example
considered in \cref{lem:strong-stronger-than-weak} provides an instance of
this. By making $\delta$ arbitrarily small, the example can be made
arbitrarily close to a QLLL problem consisting of two identical projectors
$\Pi_0$. Yet the state $\ket{\psi^\perp}$, which satisfies the strong
convergence requirements when both projectors are $\Pi_0$, is as far as
possible for satisfying those requirements for the instance given in
\cref{lem:strong-stronger-than-weak}. In order for any algorithm to determine
that $\ket{\psi^\perp}$ is \emph{not} a valid solution, the projectors would
have to be specified to arbitrarily high precision.)

To make the strong convergence case into a well-defined complexity-theoretic
problem, we must either formulate it as a \keyword{weak-membership} problem
\citep{GroetschelLovaszSchrijver}, with a parameter $\delta$ specifying the
precision to which the answer must be given, or equivalently we can formulate
it as a promise problem, with a promise on the spectral gap (cf.\
\citet{CEW09}). Either way, this means that \emph{the problem instance itself}
sets a lower-bound on the spectral gap. But \cref{lem:weak-strong-delta} shows
that an efficient algorithm for weak convergence already implies an efficient
algorithm for the complexity-theoretic formulation of strong convergence in
this case. So in the most reasonable complexity-theoretic formulation, strong
convergence and weak convergence are equivalent.

\subsection{An efficient quantum algorithm}
In \cref{sec:constructive_QLLL_proof}, we used the bound from
\cref{thm:expected_violations} together with \cref{thm:convergence_time} to
prove that \cref{alg:QLLL_converger} converges efficiently to the desired
state. But \cref{thm:convergence_time} applies to any set of projectors; it
does not depend on commutativity. Therefore, the analogous non-commutative
bounds in \cref{thm:expected_violations_poly,thm:expected_violations_exp}
(which depend on \cref{conj:cubitt-schwarz_poly,conj:cubitt-schwarz_exp})
together with \cref{thm:convergence_time} imply that \cref{alg:QLLL_converger}
converges efficiently (in the weak sense) in the non-commutative case.
\Cref{lem:weak-strong-delta} extends this to strong convergence, proving the
following constructive results in the general, non-commutative case of the
Quantum Lov\'asz Local Lemma. (The existence parts of the following theorems
are true independent of the conjectures, by \cref{thm:QLLL}.)

\begin{theorem}[Constructive QLLL]\label{thm:non-commuting_QLLL_poly}
  Let $\{\Pi_i\}$ be a set of $m$ projectors projectors acting on subsets of
  $n$ qudits, with spectral gap $\delta$ (\cref{def:spectral-gap}). If
  $\{\Pi_i\}$ satisfy the $\epsilon$-strengthened Lov\'asz conditions
  (\cref{def:strengthened_Lovasz_conditions}) for any $\epsilon > 0$, then
  there exists a joint state $\rho$ of the qudits such that $\forall i:
  \tr[\Pi_i\rho] = 0$.

  Moreover, if \cref{conj:cubitt-schwarz_poly} holds, there is a quantum
  algorithm that constructs a state $\rho'$ such that $\forall
  i\,\tr[\Pi_i\rho'] \leq \varepsilon$, in time
  \begin{equation}
    \biggOrder{n + \frac{m}{\varepsilon}\sum_{i=1}^m\frac{x_i}{1-x_i}
                   \cdot\bigAbs{[i]}}
  \end{equation}
  where $\Abs{[i]}$ is the number of qudits on which the projector $\Pi_i$
  acts non-trivially, thus $\rho'$ has probability at most $\varepsilon$ of
  violating any of the constraints defined by the $\Pi_i$.

  The same algorithm constructs a state $\rho''$ such that $\tr[P_0\rho'']
  \geq 1-\varepsilon$, where $P_0$ is the projector onto the subspace
  $\vspan\{\ket{\psi}:\forall i\,\Pi_i\ket{\psi}=0\}$, in time
  \begin{equation}
    \biggOrder{n + \frac{1}{\delta\varepsilon}\sum_{i=1}^m\frac{x_i}{1-x_i}
                   \cdot\bigAbs{[i]}},
  \end{equation}
  thus $\rho''$ is $\varepsilon$-close in fidelity to the subspace of states
  satisfying $\forall i\,\Pi_i\ket{\psi}=0$.
\end{theorem}
It may well be possible to remove the $1-\epsilon$ strengthening of the
Lov\'asz conditions, by a more refined analysis of the branching process in
\cref{thm:expected_violations_poly}. (See discussion after the proof of that
theorem.)

\begin{theorem}[Constructive QLLL, strengthened Lov\'asz conditions]
  \label{thm:non-commuting_QLLL_exp}\hfil\linebreak%
  Let $\{\Pi_i\}$ be a set of $m$ projectors projectors acting on subsets of
  $n$ qudits, with spectral gap $\delta$ (\cref{def:spectral-gap}). If
  $\{\Pi_i\}$ satisfy the $\epsilon$-strengthened Lov\'asz conditions
  (\cref{def:strengthened_Lovasz_conditions}) for $\epsilon = 1-m\delta$, then
  there exists a joint state $\rho$ of the qudits such that $\forall i\,
  \tr[\Pi_i\rho] = 0$.

  Moreover, if \cref{conj:cubitt-schwarz_exp} holds, there is a quantum
  algorithm that constructs a state $\rho'$ such that $\tr[\Pi_i\rho'] \leq
  \varepsilon$, in time
  \begin{equation}
    \biggOrder{n + \frac{m}{\varepsilon}\sum_{i=1}^m\frac{x_i}{1-x_i}
                   \cdot\bigAbs{[i]}}
  \end{equation}
  where $\Abs{[i]}$ is the number of qudits on which the projector $\Pi_i$
  acts non-trivially, thus $\rho'$ has probability at most $\varepsilon$ of
  violating any of the constraints defined by the $\Pi_i$.

  The same algorithm constructs a state $\rho''$ such that $\tr[P_0\rho'']
  \geq 1-\varepsilon$, where $P_0$ is the projector onto the subspace
  $\vspan\{\ket{\psi}:\forall i\,\Pi_i\ket{\psi}=0\}$, in time
  \begin{equation}
    \biggOrder{n + \frac{1}{\delta\varepsilon}\sum_{i=1}^m\frac{x_i}{1-x_i}
                   \cdot\bigAbs{[i]}},
  \end{equation}
  thus $\rho''$ is $\varepsilon$-close in fidelity to the subspace of states
  satisfying $\forall i\,\Pi_i\ket{\psi}=0$.
\end{theorem}

Note that all the arguments of \cref{sec:CP_map_convergence} go through for
the non-commutative case, so that
\cref{thm:non-commuting_QLLL_poly,thm:non-commuting_QLLL_exp} lead directly to
CP map convergence results for more general classes of maps. This generalises
\cref{thm:CP_map_convergence,cor:ground-state_convergence} to non-commutative
local Hamiltonians, and their associated dissipative state engineering CP
maps.

\section{Conclusions}\label{sec:conclusions}
We have proven that a simple quantum algorithm (\cref{alg:QLLL_converger})
efficiently constructs a state satisfying the requirements of the Quantum
Lov\'asz Local Lemma (\cref{thm:QLLL}) of \cite{QLLL} in the setting of
commuting constraints. Not only does this give an efficient algorithm for
constructing the quantum state whose existence is guaranteed by the
commutative QLLL. In fact, since we do not assume the QLLL in order to prove
that the algorithm finds the state, this also gives an independent,
constructive proof of the commutative QLLL itself.

Until \citet{Beck91} provided the first algorithm, it was not a priori clear
whether the combinatorial objects whose existence was guaranteed by the
classical LLL could also be constructed efficiently. In the quantum case, it
was perhaps even less obvious whether the states guaranteed to exist by the
QLLL could be prepared efficiently, as those states can have a highly complex
entanglement structure. Our result gives a new quantum algorithm for
efficiently constructing these complex states in the commutative case. Or, in
physics terms, our result provides a new method of cooling to the ground state
of certain many-body Hamiltonians with commuting local terms.

In the non-commutative setting, we can only prove the constructive QLLL modulo
a technical conjecture that the probability of witness trees is at most
polynomially weaker than the bound we have proven in the commuting case
(\cref{conj:cubitt-schwarz_poly}). If we impose stronger Lov\'asz conditions
(which reduce to the usual conditions in the commutative case), we can prove a
constructive QLLL using a weaker conjecture (\cref{conj:cubitt-schwarz_exp}),
which weakens the commutative bound by an exponential factor in the
non-commutative setting. We have given a simple counter-example showing that
the bound from the commutative case does not hold in general. The violation of
the commuting bound is directly attributable to the measurement-disturbance
effect of non-commuting quantum measurements, which means that even
``satisfied'' outcomes can disturb the state in an undesirable way.

We can prove both our conjectures in the simplest case of single-vertex trees.
This is already non-trivial, as it shows that the ``satisfied'' measurements
that can cause so much trouble in the non-commutative case can effectively be
ignored up to the first violation. However, as our counter-example
demonstrates, the case of multiple violations is substantially different, and
proving the conjectures in general is likely to require new ideas and
techniques. Our conjectures are supported by numerical evidence for some small
multi-vertex trees, though the numerics we have done are very limited.
However, it is very likely that the conjectures \emph{must} hold if the
natural quantum generalisation of Moser's algorithm (\cref{alg:quantum}) is to
work. Since the probability of witness trees occurring in the algorithm's log
is directly related to the expected number of violations seen by the algorithm
(\cref{thm:expected_violations,thm:expected_violations_poly,thm:expected_violations_exp}),
if the witness tree probability does not shrink fast enough with the size of
the tree, the expected number of violations will be unbounded. Thus disproving
our conjectures would strongly suggest that an entirely different approach to
the one pioneered by Moser in the classical setting is required in the
non-commutative quantum setting. (Or that there is no efficient constructive
version of the general QLLL.)

In the classical case, \cite{MoserTardos} impose a slight restriction on the
events that feature in the LLL, requiring that they be determined by an
underlying set of random variables. In the quantum case, we imposed the
analogous restriction, requiring that the subspaces in the QLLL are defined on
an underlying set of qudits. \cite{KashyapSzegedy} have generalised the
original \cite{MoserTardos} results, and removed this restriction in the
classical case. Their results also prove a constructive algorithm right up to
the \keyword{Shearer bound}, the tightest possible version of the Lov\'asz
conditions. Finally, they also generalise an earlier result of \cite{HSS10},
showing that the Moser algorithm is efficient in the number of variables even
if the number of events is super-polynomial. It seems likely that these
results can be generalised to the quantum setting, at least in the commutative
case. In particular, this would remove the assumption of an underlying tensor
product structure, proving a constructive version of the QLLL for general
subspaces, matching the formulation of \cite{QLLL}.

Although not expressed in this way in the published versions of the papers,
Moser pointed out that his constructive proof \citep{Moser} of the symmetric
Lov\'asz Local Lemma (\cref{cor:symmetric_LLL}) can be formulated as an
elegant compression argument \citep{Moser_notes,TerryTau_blog}: the sequence
of random bits used by the algorithm can be recovered perfectly from the data
in the execution log and final output, but if the length of the log grows
indefinitely, then this data can be compressed into fewer bits than the
entropy. The strong converse of Shannon's noiseless coding theorem (see e.g.\
\citet{Cover+Thomas}) imposes an exponential suppression of the probability of
compressing below the entropy, implying a linear bound on the expected length
of the log, hence also the run-time. This compression argument can be extended
to the general Lov\'asz Local Lemma, and it is reasonably clear
\citep{Moser_notes} that with a little effort it could even give tight
constants in the general LLL (\cref{thm:LLL}). Indeed, the witness trees and
coupling argument of \citet{MoserTardos} essentially contains an implicit
compression argument, but elegantly side-steps the necessity of designing an
explicit compression scheme for the log. Our witness tree and quantum coupling
argument again contains an implicit compression argument. But the compression
argument can also be generalised to the quantum case explicitly, either making
use of the strong converse of the Schumacher noiseless coding theorem
\citep{WinterPhD,OgawaNagaoka}, or using a more general information-theoretic
analysis~\citep{Or+Itai}.

\enlargethispage{2em}
\paragraph{Acknowledgements}
We greatfully acknowledge insightful discussions with Julia Kempe on the topic
of the QLLL, and particularly Or Sattath both for valuable discussions and for
repeatedly pointing out errors in earlier versions of our proofs. We would
also like to thank Robin Moser for taking the time to discuss his beautiful
proof with us, and for sharing valuable insights into the classical LLL, and
Itai Arad for sharing additional resuts on the QLLL. TSC would like to thank
Andreas Winter, Aram Harrow and David Perez-Garcia for valuable discussions,
James Martin for instruction in coupling arguments, and Frank Verstraete and
the University of Vienna for their hospitality throughout the visit during
which some of this work were carried out. Early parts of this work were
carried out whilst TSC was at the University of Bristol, supported by a
Leverhulme Early-Career fellowship. TSC is now supported by a Juan de la
Cierva fellowship, the EU project QUEVADIS, and by Spanish grants QUITEMAD,
I-MATH, and MTM2008-01366. MS acknowledges support by the Austrian AMS
Bildungskarenz programme and the Austrian SFB project FoQuS (F4014).

\bibliographystyle{abbrvnat}
\bibliography{QLLL}

\begin{thebibliography}{29}
\providecommand{\natexlab}[1]{#1}
\providecommand{\url}[1]{\texttt{#1}}
\expandafter\ifx\csname urlstyle\endcsname\relax
  \providecommand{\doi}[1]{doi: #1}\else
  \providecommand{\doi}{doi: \begingroup \urlstyle{rm}\Url}\fi

\bibitem[Aharonov and Eldar(2011)]{AE11}
D.~Aharonov and L.~Eldar.
\newblock On the commuting local hamiltonian problem, and tight conditions on
  topological order.
\newblock arXiv:1102.0770, 2011.

\bibitem[Aharonov et~al.(2001)Aharonov, Ambainis, Kempe, and Vazirani]{AAKV00}
D.~Aharonov, A.~Ambainis, J.~Kempe, and U.~Vazirani.
\newblock Quantum walks on graphs.
\newblock In \emph{Proceedings of the {ACM} Symposium on the Theory of
  Computation {(STOC 2001)}}, pages 50--59, July 2001.

\bibitem[Alon and Spencer(2008)]{Alon+Spencer}
N.~Alon and J.~H. Spencer.
\newblock \emph{The probabilistic method}.
\newblock Wiley, 2008.

\bibitem[Ambainis et~al.(2009)Ambainis, Kempe, and Sattath]{QLLL}
A.~Ambainis, J.~Kempe, and O.~Sattath.
\newblock A quantum {L}ov\'asz {L}ocal {L}emma.
\newblock arXiv:0911.1696[quant-ph], 2009.

\bibitem[Beck(1991)]{Beck91}
J.~Beck.
\newblock An algorithmic approach to the lov{\'a}sz local lemma.
\newblock \emph{Random Structures and Algorithms}, 2\penalty0 (4):\penalty0
  343--365, 1991.

\bibitem[Bravyi(2006)]{QSAT}
S.~Bravyi.
\newblock Efficient algorithm for a quantum analogue of {2-SAT}.
\newblock arXiv:quant-ph/0602108, 2006.

\bibitem[Bravyi and Vyalyi(2003)]{BV03}
S.~Bravyi and M.~Vyalyi.
\newblock Commutative version of the k-local hamiltonian problem and common
  eigenspace problem.
\newblock arXiv:quant-ph/0308021, 2003.

\bibitem[Cover and Thomas(2006)]{Cover+Thomas}
T.~M. Cover and J.~A. Thomas.
\newblock \emph{Elements of Information Theory}.
\newblock Wiley, second edition, 2006.

\bibitem[Cubitt et~al.(2009)Cubitt, Eisert, and Wolf]{CEW09}
T.~S. Cubitt, J.~Eisert, and M.~M. Wolf.
\newblock Deciding whether a quantum channel is {M}arkovian is {NP}-hard.
\newblock arXiv:0908.2128[math-ph], 2009.

\bibitem[Erd{\"o}s and Lov{\'a}sz(1975)]{LLL}
P.~Erd{\"o}s and L.~Lov{\'a}sz.
\newblock Problems and results on 3-chromatic hypergraphs and some related
  questions.
\newblock \emph{Infinite and finite sets}, 2:\penalty0 609--627, 1975.

\bibitem[Gr{\"o}tschel et~al.(1988)Gr{\"o}tschel, Lov{\'a}sz, and
  Schrijver]{GroetschelLovaszSchrijver}
M.~Gr{\"o}tschel, L.~Lov{\'a}sz, and A.~Schrijver.
\newblock \emph{Geometric algorithms and combinatorial optimization}.
\newblock Springer, 1988.
\newblock ISBN 038713624x.

\bibitem[Haeupler et~al.(2010)Haeupler, Saha, and Srinivasan]{HSS10}
B.~Haeupler, B.~Saha, and A.~Srinivasan.
\newblock New constructive aspects of the lovasz local lemma.
\newblock In \emph{Proceedings of the 51st Annual IEEE Symposium on Foundations
  of Computer Science (FOCS 2011)}, pages 397--406, October 2010.

\bibitem[Janzing et~al.(2003)Janzing, Wocjan, and Beth]{FBQP}
D.~Janzing, P.~Wocjan, and T.~Beth.
\newblock Cooling and low energy state preparation for 3-local {H}amiltonians
  are {FQMA}-complete.
\newblock arXiv:quant-ph/0303186, 2003.

\bibitem[Kolipaka and Szegedy(2011)]{KashyapSzegedy}
K.~Kolipaka and M.~Szegedy.
\newblock Moser and tardos meet lov{\'a}sz.
\newblock In \emph{Proceedings of the {ACM} Symposium on the Theory of
  Computation {(STOC 2011)}}, pages 235--244, June 2011.

\bibitem[Lindvall(2002)]{Lindvall}
T.~Lindvall.
\newblock \emph{Lectures on the Coupling Method}.
\newblock Dover Publications, 2002.

\bibitem[Moser and Tardos(2010)]{MoserTardos}
R.~Moser and G.~Tardos.
\newblock {A constructive proof of the general Lov{\'a}sz Local Lemma}.
\newblock \emph{Journal of the ACM (JACM)}, 57\penalty0 (2):\penalty0 1--15,
  2010.

\bibitem[Moser(2009)]{Moser}
R.~A. Moser.
\newblock A constructive proof of the {L}ov\'asz {L}ocal {L}emma.
\newblock In \emph{Proceedings of the 41st annual ACM symposium on Theory of
  computing (STOC 2009)}, 2009.

\bibitem[Moser(2010)]{Moser_notes}
R.~A. Moser.
\newblock An information theoretic view upon the constructive proof of the
  {L}ov\'asz {L}ocal {L}emma.
\newblock Private communication, 2010.

\bibitem[Nielsen and Chuang(2000)]{Nielsen+Chuang}
M.~Nielsen and I.~Chuang.
\newblock \emph{Quantum Computation and Quantum Information}.
\newblock Cambridge University Press, Cambridge, 2000.

\bibitem[Ogawa and Nagaoka(1999)]{OgawaNagaoka}
T.~Ogawa and H.~Nagaoka.
\newblock Strong converse to the quantum channel coding theorem.
\newblock \emph{IEEE Trans.\ Inform.\ Theory}, 45\penalty0 (7):\penalty0 2486,
  1999.

\bibitem[Sanz et~al.(2009)Sanz, Perez-Garcia, Wolf, and Cirac]{PWC09}
M.~Sanz, D.~Perez-Garcia, M.~M. Wolf, and J.~I. Cirac.
\newblock A quantum version of {W}ielandt's inequality.
\newblock arXiv:0909.5347[quant-ph], 2009.

\bibitem[Sattath and Arad()]{Or+Itai}
O.~Sattath and I.~Arad.
\newblock In preparation.

\bibitem[Schuch(2011)]{schuch11}
N.~Schuch.
\newblock Complexity of commuting hamiltonians on a square lattice of qubits.
\newblock arXiv:1105.2843[quant-ph], 2011.

\bibitem[Schwarz et~al.(2011)Schwarz, Temme, and Verstraete]{STV11}
M.~Schwarz, K.~Temme, and F.~Verstraete.
\newblock Contracting tensor networks and preparing peps on a quantum computer.
\newblock arXiv:1104.1410[quant-ph], 2011.

\bibitem[Tao(2009)]{TerryTau_blog}
T.~Tao.
\newblock Moser's entropy compression argument.
\newblock \url{http://terrytao.wordpress.com/2009/08/}
  \url{05/mosers-entropy-compression-argument/}, 2009.

\bibitem[Temme et~al.(2011)Temme, Osborne, , Vollbrecht, Poulin, and
  Verstraete]{TOVPF09}
K.~Temme, T.~J. Osborne, , K.~G. Vollbrecht, D.~Poulin, and F.~Verstraete.
\newblock {Quantum Metropolis Sampling}.
\newblock \emph{Nature}, 471:\penalty0 87, 2011.

\bibitem[Thorisson(2000)]{Thorisson}
H.~Thorisson.
\newblock \emph{Coupling, stationarity, and regeneration}.
\newblock Springer, 2000.

\bibitem[Verstraete et~al.(2009)Verstraete, Wolf, and Cirac]{VWC09}
F.~Verstraete, M.~M. Wolf, and J.~I. Cirac.
\newblock Quantum computation and quantum-state engineering driven by
  dissipation.
\newblock \emph{Nature Physics}, 5:\penalty0 633--636, 2009.

\bibitem[Winter(1999)]{WinterPhD}
A.~Winter.
\newblock \emph{Coding Theorems of Quantum Information Theory}.
\newblock PhD thesis, Universit\"at Bielefeld, 1999.

\end{thebibliography}

\end{document}